\documentclass[12pt]{article}
\usepackage[latin9]{inputenc}
\usepackage{geometry}
\usepackage{color}
\usepackage{mathtools}
\usepackage{amsmath}
\usepackage{amsthm}
\usepackage{amssymb}
\usepackage{stmaryrd}

\usepackage{verbatim}
\usepackage{dsfont}
\usepackage{calc}
\usepackage{mathrsfs}
\usepackage{dsfont}
\usepackage{bbm}
\usepackage{mathrsfs}
\usepackage{forest}
\usepackage{lscape}
\usepackage{tikz}

\theoremstyle{plain}
\newtheorem{thm}{\protect\theoremname}
\theoremstyle{definition}
\newtheorem{defn}[thm]{\protect\definitionname}
\theoremstyle{plain}
\newtheorem{prop}[thm]{\protect\propositionname}
\theoremstyle{plain}
\newtheorem{cor}[thm]{\protect\corollaryname}
\theoremstyle{plain}
\newtheorem{lem}[thm]{\protect\lemmaname}
\theoremstyle{remark}
\newtheorem{rem}[thm]{\protect\remarkname}

\def\mP{\mathbbm{P}}
\def\mQ{\mathbbm{Q}}
\def\mR{\mathbbm{R}}
\def\mE{\mathbbm{E}}
\def\mF{\mathbbm{F}}
\def\mG{\mathbbm{G}}

\def\mT{\mathbbm{T}}

\def\cE{\mathcal{E}}
\def\cF{\mathcal{F}}
\def\cG{\mathcal{G}}

\def\cB{\mathcal{B}}

\def\a{\alpha}
\def\b{\beta}
\def\g{\gamma}
\def\G{\Gamma}
\def\w{\omega}
\def\W{\Omega}
\def\s{\sigma}

\def\l{\lambda}

\def\p{\pi}
\def\P{\Pi}

\def\d{\delta}
\def\D{\Delta}
\def\th{\theta}

\def\e{\varepsilon}

\def\|{\lVert}

\def\inf{\text{inf }}

\def\1{\mathbbm{1}}

\def\fM{\mathfrak{M}}

\newcommand{\ov}[1]{\overline{#1}}
\newcommand{\wt}[1]{\widetilde{#1}}
\newcommand{\wh}[1]{\widehat{#1}}

\providecommand{\corollaryname}{Corollary}
\providecommand{\definitionname}{Definition}
\providecommand{\examplename}{Example}
\providecommand{\lemmaname}{Lemma}
\providecommand{\remarkname}{Remark}
\providecommand{\theoremname}{Theorem}
\providecommand{\propositionname}{Proposition}

\makeatother

\providecommand{\corollaryname}{Corollary}
\providecommand{\definitionname}{Definition}
\providecommand{\lemmaname}{Lemma}
\providecommand{\propositionname}{Proposition}
\providecommand{\remarkname}{Remark}
\providecommand{\theoremname}{Theorem}

\begin{document}

\title{Filtration Reduction and Completeness in Jump-Diffusion Models}

\author{\textcolor{black}{Karen Grigorian}\thanks{\textcolor{black}{Operations Research and Information Engineering,
Cornell University, Ithaca, N.Y. 14853. Currently: Department of Statistics and Applied Probability, University of California, Santa Barbara, CA, 93101. Email: grigorian@ucsb.edu}}\textcolor{black}{{} \ and Robert A. Jarrow}\thanks{\textcolor{black}{Samuel Curtis Johnson Graduate School of Management,
Cornell University, Ithaca, N.Y. 14853. email: raj15@cornell.edu.}}}

\date{}

\maketitle

\begin{abstract}
This paper studies the pricing and hedging of derivatives
in frictionless and competitive, but incomplete jump-diffusion markets.
A unique equivalent martingale measure (EMM) is obtained using filtration
reduction to a fictitious complete market. This unique EMM in the
fictitious market is uplifted to the original economy using the notion
of consistency. For pedagogical purposes, we begin with simple setups
and progressively extend to models of increasing generality.
\end{abstract}


\section{Introduction}

\textcolor{black}{This paper studies the pricing of derivatives in
frictionless, competitive, and arbitrage-free but incomplete jump-diffusion
markets. A unique price for a derivative in this setting is obtained
using the new uplifted equivalent martingale measure (EMM) methodology
developed by Grigorian and Jarrow 2023 \cite{key-12,key-13,key-14, key-17}
in a sequence of }\textcolor{black}{papers. This}\textcolor{black}{{}
methodology uses filtration reduction to construct a fictitious, but
complete market, where a unique EMM is determined. The reduced filtration
is selected by the trader so that the risks the filtration generates
can be completely hedged. Then, the EMM is uplifted to a unique EMM
in the original incomplete market satisfying a consistency condition,
which implies that the remaining non-hedged risks have no risk premium.}



\textcolor{black}{Our paper is related to a large literature on the
pricing of derivatives in incomplete, but arbitrage-}\textcolor{black}{free
markets. See}\textcolor{black}{{} Bingham and Kiesel 2000 \cite{Bingham Kiesel 2000},
Chapter 7, for an excellent summary of this literature. As is well-known
in this literature, in an incomplete market, there is no unique price
for a derivative whose payoff is non-linear in the traded assets.
This is because the set of }\textcolor{black}{EMMs }\textcolor{black}{contains
a continuum of elements. Various methods have been employed to select
a unique element from this set.}

\textcolor{black}{One approach is to choose a price implied by a given
objective }\textcolor{black}{function, which}\textcolor{black}{{} includes
variance and risk-minimizing hedging and indifference pricing. Schweizer
1994 \cite{Schweizer} constructs risk-minimizing hedging strategies
under restrictions on the available information. }A second approach
is to choose a unique\textcolor{black}{{} EMM} determined by some economic
justification or an analytic condition. Examples of choosing a unique
martingale measure based on the economic justification that certain
risks are diversifiable are the papers \textcolor{black}{by Merton 1976
\cite{Merton 1976} who} assumes jump risk and Hull and White 1987
\cite{Hull White 1987} \textcolor{black}{who assumes} volatility risk
are diversifiable, which implies a unique EMM. The minimal martingale
measure is the prime example of choosing a unique probability that
satisfies an analytic condition. The minimal martingale measure also
has the property that all non-traded risks (those in the larger filtration
$\mathbbm{G}$ but not in $\mathbbm{F}$) are not priced (see Bingham
and Kiesel 2000 \cite{Bingham Kiesel 2000}, p. 239). The problem
with both of these approaches in practice is that the diversifiable
risk assumption does not always apply, and choosing an appropriate
objective function is often difficult.

A somewhat different strand of literature focuses on the mathematical
nuances of optional projections, such as the loss of the local martingale
property in a reduced filtration, see F\"ollmer and Protter 2011 \cite{Follmer Protter}
and Larsson 2014 \cite{Larsson}. More recently, Bielecki et al 2021
\cite{Bielecki} investigated the semimartingale characteristics of
special semiamrtingales under filtration shrinkage. Biagini et all
2023 \cite{Biagini} study optional projections of strict $\mG$-local
martingales onto smaller filtrations $\mF$ under changes of equivalent
martingale measures and the absence of arbitrage opportunities.

Closer to this paper is Cuchiero et al 2020 \cite{Cuchiero} who provide
a version of the fundamental theorem of asset pricing for continuous-time
large financial markets with two filtrations, extending the results
previously obtained by Kabanov and Stricker 2006 \cite{Kabanov Stricker}
in a discrete-time setting. A slightly different problem is analyzed
by Kardaras and Ruf 2019 \cite{key-3}, who study the relation between
filtration shrinkage and local martingale deflators.

Our method differs from the previous literature in that we select
a unique \textcolor{black}{EMM} determined by the collection of sub-filtrations
$\mathbbm{F}$ for which the implied fictitious market is complete.
This equivalent probability, which we call the uplifted \textcolor{black}{EMM},
also has the property that all non-hedged risks (those in $\mathbbm{G}$
but not in $\mathbbm{F}$) are non-priced. In contrast to the aforementioned
literature, our unique pricing is explicitly determined by the risks
(filtration) that a trader desires to hedge, knowing in advance that
not all the risks in market prices can be hedged. Our method invokes
no explicit assumptions on preferences nor analytically convenient
properties of martingale processes. As such, it can be readily applied
in practice \textcolor{black}{to price } traded derivatives.

\textcolor{black}{Given }\textcolor{black}{that}\textcolor{black}{{} the
economic motivation underlying the uplifted EMM methodology is detailed
in Grigorian and Jarrow 2023 \cite{key-12,key-13,key-14, key-17}, }\textcolor{black}{we
concentrate here}\textcolor{black}{{} on the mathematics }\textcolor{black}{underlying}\textcolor{black}{{}
the application. An outline of this paper is as follows. Section 2
presents the preliminaries, key theorems used throughout the paper.
Section 3 studies filtration reduction with complete neglect, and
section 4 studies filtration reduction with partial neglect. Section
5 discusses }\textcolor{black}{the }\textcolor{black}{pricing and hedging
of derivatives, while section 6 concludes.}

\section{Preliminaries}

\textcolor{black}{This section collects some basic facts from the
theory of stochastic processes with }\textcolor{black}{jumps utilized
}\textcolor{black}{in the subsequent sections. }\textcolor{black}{Although
some of this material is basic,
these results allow us to}\textcolor{black}{{} present the key ideas
in }\textcolor{black}{a form}\textcolor{black}{{} that facilitates the
understanding of the }\textcolor{black}{more complex}\textcolor{black}{{}
cases that follow. We assume familiarity with the notions of counting processes, the Poisson process, the compound Poisson process and the standard Brownian motion.}

The following decomposition \textcolor{black}{proposition} and its corollary are important in the subsequent discussion. The proofs can be found in Shreve, 2004 \cite{key-16}, Ch. 11, the notation of which we follow closely \textcolor{black}{(see also }Chin et al, 2014 \cite{key-10}).
\begin{prop}
Let $y_{1},...,y_{M}$ be a finite set of nonzero numbers, and let
$p(y_{1}),...,p(y_{M})$ be positive numbers \textcolor{black}{summing}
to 1. \textcolor{black}{Given} \textcolor{black}{the constant }$\l>0$,
let $\ov{N}_{1},...,\ov{N}_{M}$ be independent Poisson processes,
each $\ov{N}_{m}$ having intensity $\l p(y_{m})$. Define $\ov{Q}(t)=\sum_{m=1}^{M}y_{m}\ov{N}_{m}(t),\quad t\geq0.$ Then, $\ov{Q}$ is a compound Poisson process. \textcolor{black}{In particular}, if $\ov{Y}_{1}$ is the size of the first jump of $\ov{Q}$, $\ov{Y}_{2}$ is the size of the second jump\textcolor{black}{,...,
etc.}, and $\ov{N}(t)=\sum_{m=1}^{M}\ov{N}_{m}(t),\quad t\geq0,$ is the \textcolor{black}{total }number of jumps on the time interval
$(0,t]$, then $\ov{N}$ is a Poisson process with intensity $\l$,
the random variables $\ov{Y}_{1},\ov{Y}_{2},...$ are independent
with $\mP(\ov{Y}_{i}=y_{m})=p(y_{m})$ for $m=1,...,M$, the random
variables $\ov{Y}_{1},\ov{Y}_{2},...$ are also independent of $\ov{N}$,
and $\ov{Q}(t)=\sum_{i=1}^{\ov{N}(t)}\ov{Y}_{i},\quad t\geq0.$
\end{prop}

\begin{cor}
Let $y_{1},...,y_{M}$ be a finite set of nonzero numbers, and let
$p(y_{1}),...,p(y_{M})$ be positive numbers \textcolor{black}{summing}
to 1. Let $Y_{1},Y_{2},...$ be a sequence of \textcolor{black}{i.i.d.}
random variables with $\mP(Y_{i}=y_{m})=p(y_{m})$ for $m=1,...,M$.
Let $N$ be a Poisson process and define the compound Poisson process $Q(t)=\sum_{i=1}^{N(t)}Y_{i}.$ For $m=1,...,M,$ let $N_{m}(t)$ denote the number of jumps in $Q$ of size $y_{m}$ up to and including time $t$. Then, $N(t)=\sum_{m=1}^{M}N_{m}(t)\quad\text{and}\quad Q(t)=\sum_{m=1}^{M}y_{m}N_{m}(t),$ \textcolor{black}{the} processes $N_{1},...,N_{M}$ are independent
Poisson processes, and each $N_{m}$ has intensity $\l p(y_{m})$. 
\end{cor}

\begin{prop}
Let $W$ be a Brownian motion and $N$ a Poisson process with \textcolor{black}{constant}
intensity $\l>0$, both defined on the same probability space $(\W,\cF,\mF,\mP)$
and $\mF$-adapted. Then, the processes $W$ and $N$ are independent. 
\end{prop}


\textcolor{black}{We also need results on changes} of measures
for Poisson processes, compound Poisson processes, and combinations
with Brownian motions. Let $M(t)=N(t)-\l t$ be a compensated Poisson
process on $(\W,\cF,\mF,\mP)$, which is a martingale. Let $\wt{\l}$
be a positive number and define $Z(t)=e^{(\l-\wt{\l})t}\left(\frac{\wt{\l}}{\l}\right)^{N(t)}.$ For a fixed $T>0$, we would like to use $Z(T)>0$ to change to a
new measure $\wt{\mP}$ under which $(N(t),0\leq t\leq T)$ has intensity
$\wt{\l}$ instead of $\l$. Note that $Z$ is a solution to the stochastic differential equation $dZ(t)=\frac{\wt{\l}-\l}{\l}Z(t-)dM(t),$ and is a martingale under $\mP$ with $\mE Z(t)=1$ for all $t.$
Define the new measure via $\wt{\mP}(A)=\int_{A}Z(T)d\mP\quad\text{for all }A\in\cF.$

\begin{prop}
\textbf{\textcolor{black}{(Change of Poisson Intensity)}}

Under the probability measure $\wt{\mP}$, the process $(N(t),0\leq t\leq T)$
is \textcolor{black}{a} Poisson\textcolor{black}{{} process }with intensity
$\wt{\l}$. 
\end{prop}

Similarly, for positive numbers $\wt{\l}_{1},...,\wt{\l}_{M}$, the
process $Z$, defined via $Z_{m}(t)=e^{(\l_{m}-\wt{\l}_{m})t}\left(\frac{\wt{\l}_{m}}{\l_{m}}\right)^{N_{m}(t)},\quad Z(t)=\prod_{m=1}^{M}Z_{m}(t),$ is a martingale under $\mP$ with $\mE Z(t)=1$ for all $t$, and
can be used to define $\wt{\mP}$ as above. 
\begin{prop}
\textbf{(Change of Compound Poisson Intensity and Jump Distribution
for a Finite Number of Jump Sizes)}

Under $\wt{\mP}$, $Q$ is a compound Poisson process with intensity
$\l=\sum_{m=1}^{M}\wt{\l}_{m}$, and $Y_{1},Y_{2},...$ are \textcolor{black}{i.i.d.}
random variables with 
\begin{align*}
\wt{\mP}(Y_{i}=y_{m})=\wt{p}(y_{m})=\frac{\wt{\l}_{m}}{\wt{\l}}.
\end{align*}
\end{prop}

Writing the Radon-Nikodym derivative as 
\begin{align*}
Z(t)=\exp\left(\sum_{m=1}^{M}(\l_{m}-\wt{\l}_{m})t\right)\prod_{m=1}^{M}\left(\frac{\wt{\l}\wt{p}(y_{m})}{\l p(y_{m})}\right)^{N_{m}(t)}=e^{(\l-\wt{\l})t}\prod_{i=1}^{N(t)}\frac{\wt{\l}\wt{p}(Y_{i})}{\l p(Y_{i})},
\end{align*}
\textcolor{black}{this approach} can be extended to jump sizes $Y_{i}$\textcolor{black}{{}
with density }$f(y)$ via 
\begin{align*}
Z(t)=e^{(\l-\wt{\l})t}\prod_{i=1}^{N(t)}\frac{\wt{\l}\wt{f}(Y_{i})}{\l f(Y_{i})},
\end{align*}
where we assume that $\wt{f}(y)=0$ whenever $f(y)=0$ to avoid division
by 0 and to ensure that $\wt{\mP}\sim\mP$. This process $Z$ is also
a martingale under $\mP$ with $\mE Z(t)=1$ for all $t$, and hence
we can define $\wt{\mP}$ the usual way. 
\begin{prop}
\textbf{(Change of Compound Poisson Intensity and Jump Distribution
for a Continuum of Jump Sizes)}

Under $\wt{\mP}$, the process $Q$ is a compound Poisson process
with \textcolor{black}{constant} intensity $\wt{\l}$ and the jumps
in $Q$ are \textcolor{black}{i.i.d.} with density $\wt{f}(y)$. 
\end{prop}

Finally, we provide a change of measure formula for a model driven
by a Brownian motion and Poisson processes. Specifically, let $\th(t)$
be an adapted process and define 
\begin{align*}
Z_{1}(t) & =\exp\left(-\int_{0}^{t}\th(u)dW(u)-\frac{1}{2}\int_{0}^{t}\th^{2}(u)du\right),\\
Z_{2}(t) & =e^{(\l-\wt{\l})t}\prod_{i=1}^{N(t)}\frac{\wt{\l}\wt{f}(Y_{i})}{\l f(Y_{i})},\\
Z(t) & =Z_{1}(t)Z_{2}(t),
\end{align*}
which, under some conditions on $\theta(t)$, can be shown to be a martingale under $\mP$ with $\mE Z(t)=1$
for all $t$. 
\begin{prop}
Under the probability measure $\wt{\mP}$, the process 
\begin{align*}
\wt{W}(t)=W(t)+\int_{0}^{t}\th(s)ds
\end{align*}
is a Brownian motion, $Q$ is a compound Poisson process with intensity
$\wt{\l}$, i.i.d. jump sizes with density $\wt{f}(y)$, and the processes
$\wt{W}$ and $Q$ are independent. 
\end{prop}





\begin{defn}
(\textcolor{black}{\emph{M-dimensional Counting and Point Processes}})
\textcolor{black}{A }\textbf{$M$-dimensional counting process} is a
vector $(N^{1},...,N^{M})$ of counting processes. \textcolor{black}{A}
\textbf{$M$-variate point process} is an $M$-dimensional counting
process without common jumps, i.e. $\D N^{i}(t)\D N^{j}(t)=0$ for
all $t$ a.s. 
\end{defn}

If $(N^{1},...,N^{M})$ is a multivariate point process, then we can
define a process $N$ by $N(t)=\sum_{i=1}^{M}N^{i}(t).$ The absence of common jumps implies that $N$ is a counting process. Denoting by $\{T_{n}\}$ the jump times of $N$, we define a sequence of random variables taking values in $\{1,...,M\}$ by $Y_{n}=\sum_{i=1}^{M}i\D N^{i}(T_{n}),$ i.e. if the process $N^{i}$ has a jump at $T_{n}$, then $Y_{n}=i$. With this setup, the processes $N^{1},...,N^{M}$ can be recovered from $N$ and $\{Y_{n}\}$ via $N^{i}(t)=\sum_{n=1}^{\infty}\1_{\{T_{n}\leq t\}}\1_{\{Y_{n}=i\}}.$

Thus, \textcolor{black}{a} $M$-variate point process can be described
either by the vector process $(N^{1},...,N^{M})$, or by the sequence
$\{(T_{n},Y_{n})\}_{n=1}^{\infty}$, i.e. the paths $t\mapsto(N^{1}(\w),...,N^{M}(\w))$
are uniquely determined by the collection of points $\{(T_{n}(\w),Y_{n}(\w)\}_{n=1}^{\infty}$.

Given a counting process $N$ and a sequence of random variables $\{Y_{n}\}$
taking values in $\{1,...,M\}$, define for every $i=1,...,M,$ the counting process $N^{i}$ by $N^{i}(t)=\sum_{n=1}^{\infty}\1_{\{T_{n}\leq t\}}\1_{\{Y_{n}=i\}}.$ The multivariate process $N^{Y}\coloneqq(N^{1},...,N^{M})$ is the
said to be obtained by \textbf{$Y$-marking}. We have the following
result, a version of which we have encountered before. 
\begin{prop}
Assume $N^{Y}\coloneqq(N^{1},...,N^{M})$ is obtained from a Poisson
process $N$ with intensity function $\l(t)$, where $\{Y_{n}\}$
are\textcolor{black}{{} i.i.d.} and also independent of $N$, where
$\mP(Y_{n}=i)=p_{i}$. \textcolor{black}{Then,}

(i) the process $N^{i}$ is $\mF^{N^{Y}}$-Poisson with intensity
$p_{i}\l(t)$ for all $i=1,...,M$;

(ii) the processes $N^{1},...,N^{M}$ are independent. 
\end{prop}

\begin{proof}
See, for example, Bj\"ork, 2021 \cite{key-9}, Proposition 3.8, p.35. 
\end{proof}
Finally, consider \textcolor{black}{a} $M$-variate point process $N^{1},...,N^{M}$
with intensities $\l^{1},...,\l^{M}$. Define $\l(t)=\sum_{i=1}^{M}\l^{i}(t)$
and $N(t)=\sum_{i=1}^{M}N^{i}(t)$. Then, following \textcolor{black}{the
heuristic presentation} in Bj\"ork, 2021 \cite{key-9}, p.41, given
information on $N$ over $[0,t)$ and the fact that there is a jump
of $N$ at time $t$, the conditional probability\textcolor{black}{{}
of the} process $N^{i}$ \textcolor{black}{generating} the jump is 
\begin{align*}
\mP(dN^{i}(t)=1|\cF_{t-}^{N},dN(t)=1) & =\frac{\mP(dN^{i}(t)=1,dN(t)=1|\cF_{t-}^{N})}{\mP(dN(t)=1|\cF_{t-}^{N})}\\
 & =\frac{\mP(dN^{i}(t)=1|\cF_{t-}^{N})}{\l(t)dt}=\frac{\l^{i}(t)dt}{\l(t)dt}=\frac{\l^{i}(t)}{\l(t)}.
\end{align*}
This can be rigorously stated as (Last and Brandt, 1995 \cite{key-15}): 
\begin{prop}
Let $Y_{n}=\sum_{i=1}^{M}i\D N^{i}(T_{n})$. Then, 
\begin{align*}
\mP(Y_{n}=i|\cF_{T_{n}-}^{N})=\frac{\l^{i}(T_{n})}{\l(T_{n})}\quad\text{on the set}\quad\{T_{n}<\infty\}.
\end{align*}
\end{prop}

All of the above \textcolor{black}{propositions }can \textcolor{black}{be
generalized} \textcolor{black}{using }marked point processes and \textbf{random
measures}. We will \textcolor{black}{follow} Bj\"ork, 2021 \cite{key-9}\textcolor{black}{,
although} \textcolor{black}{more} detailed treatments are available
(\textcolor{black}{for example,} Last and Brandt, 1995 \cite{key-15}).

Consider a filtered probability space $(\W,\cF,\mF,\mP)$ and a measurable
mark space $(E,\cE)$. 
\begin{defn}
(\textcolor{black}{\emph{Market Point Process}}) A \textbf{marked
point process} $\Psi$ is a random measure on $(\mR_{+}\times E,\cB(\mR_{+})\otimes\cE)$. 
\end{defn}

$\Psi$ is a counting measure (\textcolor{black}{takes }only integer
values)\textcolor{black}{, i.e. }for every fixed $t$, $\Psi(\{t\},E)=0$
or $1$\textcolor{black}{, moreover,} $\Psi([0,t],E)<\infty$ $\mP$-a.s.
\textcolor{black}{For }every $A\in\cE$, we define the counting process
$N^{A}(t)\coloneqq\Psi([0,t]\times A)$. As was\textcolor{black}{{}
previously} mentioned, $\Psi$ can \textcolor{black}{be equivalently}
characterized in terms of the sequence $\{(T_{n},Y_{n})\}_{n=1}^{\infty}$,
where $(T_{n},Y_{n})\in\mR_{+}\times E$ and $T_{1}<T_{2}<...$. \textcolor{black}{Here,}
$\{T_{n}\}_{n}^{\infty}$ is the sequence of jump times (also called
event times) for the counting process $N^{E}(t),$ i.e. $T_{n}=\inf\{t\geq0:\Psi([0,t]\times E)=n\}.$

For a given marked point process $\Psi$ and a stochastic process
$\b:\mR_{+}\times E\times\W\mapsto\mR$, the stochastic integral \textcolor{black}{is
defined b}y $X(t)=\int_{0}^{t}\int_{E}\b_{s}(y)\Psi(dy,ds),$ or, in differential form $dX(t)=\int_{E}\b_{t}(y)\Psi(dy,dt),$ where $\b_{t}(y)$ \textcolor{black}{is the} jump size given the even
$\{T_{n}=t\}$. This stochastic integral \textcolor{black}{can be rewritten
as} the process $X(t)=\sum_{T_{n}\leq t}\b_{T_{n}}(Y_{n}).$

\begin{defn}
(\textcolor{black}{\emph{A Predictable Intensity}}) \textcolor{black}{The}
marked point process $\Psi$ \textcolor{black}{has} a \textbf{predictable
intensity} $\l:\mR_{+}\times\cE\times\W\mapsto\mR$ if \\
 (i) for each $(t,\w)$, $\l_{t}(\cdot,\w)$ is a nonnegative measure
on $(E,\cE)$; \\
 (ii) for each $A\in\cE$, the process $t\mapsto\l_{t}(A)$ is $\mF$-predictable;
\\
 (iii) for each $A\in\cE$, the counting process $N^{A}(t)\coloneqq\Psi([0,t]\times A)$
has the $\mF$-predictable intensity process $\l(A)$. 
\end{defn}

When $\l$ \textcolor{black}{exists, }\textcolor{black}{the compensated} martingale process\textcolor{black}{{} is} $N^{A}(t)-\int_{0}^{t}\l_{s}(A)ds.$ The heuristic interpretation of $\l$ as $\l_{t}(A)dt=\mE(dN^{A}(t)|\cF_{t-})$
still holds. More formally, for any $A\in\cE$ 
\begin{align*}
\mP(Y_{n}\in A|\cF_{T_{n}-}^{\Psi})=\frac{\l_{T_{n}}(A)}{\l_{T_{n}}(E)},\quad\text{on the set}\quad\{T_{n}<\infty\}.
\end{align*}
Hence, the intensity process $\l$ generates a measure-valued process $\G_{t}(dy)=\frac{\l_{t}(dy)}{\l_{t}(E)},$ and the marked point process can be described \textcolor{black}{via:
}(i) the intensity process $\l_{t}(E)$ that governs the distribution
of the event times $\{T_{n}\}_{n}^{\infty}$, and (ii) \textcolor{black}{given
an }event time $t$, the probability measure $\G_{t}(dy)$ \textcolor{black}{that
}governs the distribution over the mark space. 
\begin{prop}
\label{determ is poisson} Consider a marked point process $\Psi$
with intensity process $\l$. If $\l_{t}(dy)$ is deterministic, then
\\
 (i) the counting process $N^{E}(t)$ is Poisson with intensity function
$\l_{t}(E)$; \\
 (ii) for any collection $A_{1},...,A_{M}$ of $\cE$-measurable disjoint
subsets of $E$, the counting processes $N^{A_{1}},...,N^{A_{M}}$
are independent Poisson processes with intensity functions $\l_{t}(A_{1}),...,\l_{t}(A_{M})$,
respectively. 
\end{prop}

\textcolor{black}{When} $\l_{t}(dz)$ is \textcolor{black}{a constant,}
the intensity is a fixed nonrandom measure $\l(dy)$, \textcolor{black}{and}
the \textcolor{black}{events follow} a Poisson process with constant
intensity $\l(E)$\textcolor{black}{. When} an event occurs, the value
of $Y$ is determined by the probability measure $\G(dy)=\frac{\l(dy)}{\l(E)}$.
\textcolor{black}{Here,} $\Psi$ is a compound Poisson process.

Assume there are two marked point processes $\Psi_{1}$ and $\Psi_{2}$
defined on the same probability space with intensity measures $\l_{t}^{1}(dy)$
and $\l_{t}^{2}(dy)$\textcolor{black}{, respectively.} 
\begin{prop}
\label{no common jumps} Assume that $\Psi_{1}$ and $\Psi_{2}$ have
no common event times. Then, the random measure $\Psi$, defined by $\Psi(dt,dy)=\Psi_{1}(dt,dy)+\Psi_{2}(dt,dy),$ is a marked point process with intensity measure $\l_{t}(dy)=\l_{t}^{1}(dy)+\l_{t}^{2}(dy).$
\end{prop}

\textcolor{black}{If the integrand of the stochastic integral with this
intensity measure satisfies suitable}\textcolor{black}{{} i}ntegrability
properties, the \textcolor{black}{integral is} a (local) martingale.
More precisely, 
\begin{prop}
Let $\Psi$ be a marked point process with predictable intensity measure
$\l$. Assume that $H$ is predictable \textcolor{black}{with} $\mE\left(\int_{0}^{\infty}\int_{E}|H_{t}(y)|\l_{t}(dy)dt\right)<\infty.$ Then, the process $M$ defined by $M(t)=\int_{0}^{t}\int_{E}H_{s}(y)\{\Psi(ds,dy)-\l_{s}(dy)ds\}$ is a martingale. If $H$ only satisfies $\int_{0}^{\infty}\int_{E}|H_{t}(y)|\l_{t}(dy)dt<\infty\quad\mP\text{-a.s.},$ then $M$ is a local martingale. 
\end{prop}

\begin{prop}
\textbf{(A Martingale Representation Theorem)}
Consider a filtered probability space $(\W,\cF,\mF,\mP)$ that carries a marked point process $\Psi$ with mark space $E$ \textcolor{black}{and
a} $\mF$-predictable intensity $\l$. Assume that $\mF$ is the natural filtration generated by $\Psi$ and let $X$ be a $(\mP,\mF)$-martingale (or local martingale). Then, there exists a predictable process $H$ and a real number $x$ such that $X(t)=x+\int_{0}^{t}\int_{E}H_{s}(y)\{\Psi(ds,dy)-\l_{s}(dy)ds\},$ or, in differential form, $dX(t)=\int_{E}H_{t}(y)\{\Psi(dt,dy)-\l_{t}(dy)dt\},$ where $H$ satisfies the corresponding conditions above in the martingale case and in the local martingale case. 
\end{prop}

Finally, we will need the \textcolor{black}{following version }\textcolor{black}{of
Girsanov's} theorem for marked point processes. 
\begin{prop}
\textbf{(Girsanov's Theorem for Point Processes)}

Consider a filtered probability space $(\W,\cF,\mF,\mP)$ that carries
a marked point process $\Psi$ with mark space $E$ \textcolor{black}{and
a} $\mF$-predictable intensity $\l$. Let $\varphi$ be a predictable process satisfying $\varphi_{t}(y)\geq0,\quad\mP-\text{a.s. \quad for all}\quad(t,y).$ \textcolor{black}{Define the} Radon-Nikodym process $Z$ by $dZ(t) =Z(t-)\int_{E}(\varphi_{t}(y)-1)\left[\Psi(dy,dt)-\l_{t}(dy)dt\right], Z(0) =1.$ Assume that $\mE_{\mP}Z(T)=1$ and define the probability measure
$\wt{\mP}$ by $Z(T)=\frac{d\wt{\mP}}{d\mP}\quad\text{on}\ \cF_{T}.$ Then, $\Psi$ has the predictable $(\wt{\mP},\mF)$-intensity $\wt{\l}$
given by $\wt{\l}_{t}(dy)=\varphi_{t}(y)\l_{t}(dy).$
\end{prop}

We end with a version of \textcolor{black}{Girsanov's }theorem that
includes a $D$-dimensional Brownian motion process $W.$ 
\begin{prop}
\textbf{(Girsanov's Theorem for General Jump-Diffusions)}

Consider a filtered probability space $(\W,\cF,\mF,\mP)$ that carries
a marked point process $\Psi$ with mark space $E$\textcolor{black}{,
a} $\mF$-predictable intensity $\l$\textcolor{black}{, and} a $D$-dimensional
Brownian motion $W$. Let $\th$ be a $D$-dimensional column vector process and let $\varphi$ be a predictable process satisfying $\varphi_{t}(y)\geq0,\quad\mP-\text{a.s. \quad for all}\quad(t,y),$ \textcolor{black}{Define the} Radon-Nikodym process $Z$ by 
\begin{align*}
dZ(t) & =Z(t)\th(t)^{T}dW(t)+Z(t-)\int_{E}(\varphi_{t}(y)-1)\left[\Psi(dy,dt)-\l_{t}(dy)dt\right].\\
Z(0) & =1,
\end{align*}
where $^{T}$ denotes transpose. Assume that $\mE_{\mP}Z(T)=1$ and
define the probability measure $\wt{\mP}$ by $Z(T)=\frac{d\wt{\mP}}{d\mP}\quad\text{on}\ \cF_{T}.$ Then, $\Psi$ has the predictable $(\wt{\mP},\mF)$-intensity $\wt{\l}$ given by $\wt{\l}_{t}(dy)=\varphi_{t}(y)\l_{t}(dy),$ and $\wt{W}$, defined by $d\wt{W}(t)=dW(t)-\th(t)dt,$ is a $\wt{\mP}$- Brownian motion. 
\end{prop}

\textcolor{black}{The first fundamental theorem of asset pricing states that for a nonnegative semimartingale $S$, no arbitrage in the sense of no free lunch with vanishing risk (NFLVR, see Jarrow \cite{key-11}) is equivalent
to the existence of a probability measure $\mQ$ on $(\Omega,\mathcal{G})$ equivalent to $\mP$ under which $S$ is a local martingale. We denote the set of equivalent martingale measures
(EMMs) for $S$ w.r.t. $\mG$ by $\mathfrak{\mathfrak{M}}(S,\mG)$. Similarly, the set of equivalent local
martingale measures (ELMMs) will be denoted as $\mathfrak{M}_{l}(S,\mG).$ We drop the dependence on
$S$ and the underlying filtration $\mG$ when there is no confusion.
\begin{prop}
\textbf{(The First Fundamental Theorem }\textbf{\textcolor{black}{(FFTAP)}}\textbf{)}
NFLVR if and only if $\fM_{l}\neq\varnothing$.
\end{prop}}

\begin{defn}
Assume $\mathfrak{M}_{l}\neq\emptyset$. Choose $\mQ\in\mathfrak{M}_{l}$. The market is said to be complete with respect to $\mQ$ if given any $Z_{T}\in L_{+}^{1}(\mQ)$, there exists an $x\geq0$ and $(\alpha_{0},\alpha)\in\mathcal{A}(x)$ (the set of admissible self-financing trading strategies (s.f.t.s.), where $x$ is the initial wealth, see Jarrow \cite{key-11}) such that $x+\int_{0}^{T}\alpha_{u}\cdot dS_{u}=Z_{T}$
and the value process defined by $X_{t}\coloneqq x+\int_{0}^{t}\alpha_{u}\cdot dS_{u}$
for all $t\in[0,T]$ is a $\mQ$-martingale.
\end{defn}

\textcolor{black}{
\begin{prop}
\textbf{(The Second Fundamental Theorem }\textbf{\textcolor{black}{(SFTAP)}}\textbf{)}
Assume $\fM_{l}\neq\varnothing$. If the market is complete
with respect to $\mQ\in\fM_{l}$, then $\fM_{l}$ is a singleton.
Conversely, if $\fM\neq\varnothing$, and $\fM$ is a singleton, then the market is
complete with respect to $\mQ\in\fM$.
\end{prop}}

\section{Filtration Reduction with Complete Neglect}

\textcolor{black}{This section illustrates the uplifted EMM methodology
for a filtration reduction with complete neglect}. \textcolor{black}{The
terminology ``complete neglect'' is explained below.}

\subsection{\textcolor{black}{Three Stocks, One Brownian Motion, and Three Poisson
Processes}}

\textcolor{black}{We begin with }\textcolor{black}{a constant}\textcolor{black}{{}
coefficient discrete jump process that captures the main ideas and
allows to}\textcolor{black}{{} present the the key insights in a transparent
fashion.} \textcolor{black}{We generalize}\textcolor{black}{{} in subsequent
sections.}

\subsubsection{The Original Market}

\textcolor{black}{This section introduces the original market. }We are
given a complete filtered probability space $(\W,\cG,\mG,\mP)$ that
satisfies the usual hypotheses and supports a Brownian motion $W$
and $M$ independent Poisson processes $N_{1},...,N_{M}$, all on
a finite horizon $0\leq t\leq T$.

For $\l_{m}>0$ and $-1<y_{1}<...<y_{M}$, define $N(t)=\sum_{m=1}^{M}N_{m}(t),\quad Q(t)=\sum_{m=1}^{M}y_{m}N_{m}(t).$ Thus, $N$ is a Poisson process with intensity $\l=\sum_{m=1}^{M}\l_{m}$
and $Q$ is a compound Poisson process. Denote by $Y_{i}$ the size
of the $i$th jump of $Q$. \textcolor{black}{The distribution} of
$Y_{i}$ is supported on $\{y_{1},...,y_{M}\}$ and $Q$ can be written
as $Q(t)=\sum_{i=1}^{N(t)}Y_{i}.$ Setting $p(y_{m})=\frac{\l_{m}}{\l}$, it follows that \textcolor{black}{the
jumps} $Y_{i}$ are \textcolor{black}{i.i.d.} with $p(y_{m})=\mP(Y_{i}=y_{m})$.
Set $\b\coloneqq\mE Y_{i}=\sum_{m=1}^{M}y_{m}p(y_{m})=\frac{1}{\l}\sum_{m=1}^{M}\l_{m}y_{m}.$ Note that $Q(t)-\b\l t$ is \textcolor{black}{a $\mathbbm{P}$}-martingale.

\textcolor{black}{For the moment, consider only a single stock trading
}\textcolor{black}{in the original market }\textcolor{black}{with the
price process} $dS(t)=(\a-\b\l)S(t)dt+\s S(t)dW(t)+S(t-)dQ(t).$ The solution to this stochastic differential equation is given by $S(t)=S(0)\exp\{\s W(t)+(\a-\b\l-\frac{1}{2}\s^{2})t\}\prod_{i=1}^{N(t)}(Y_{i}+1),$ see Shreve, 2004 \cite{key-16}, Theorem 11.7.3, p.513.

\subsubsection{No-Arbitrage \textcolor{black}{and An Incomplete Market}}

\textcolor{black}{An EMM in the original market is co}nstructed as
follows. For a constant $\th$ and positive constants $\wt{\l}_{1},...,\wt{\l}_{M}$,
we define $Z_{0}(t)=\exp\{-\th W(t)-\frac{1}{2}\th^{2}t\},
Z_{m}(t)=e^{(\l_{m}-\wt{\l}_{m})t}\left(\frac{\wt{\l}_{m}}{\l_{m}}\right)^{N_{m}(t)},m=1,...,M,
Z(t)=Z_{0}(t)\prod_{m=1}^{M}Z_{m}(t),
\mQ(A) =\int_{A}Z(T)d\mP\quad\text{for all}\quad A\in\cG.$ \textcolor{black}{Then, it }follows that under $\mQ$, \\
 (i) the process $\wt{W}(t)=W(t)-\th t$ is a Brownian motion, \\
 (ii) each $N_{m}$ is a Poisson process with intensity $\wt{\l}_{m}$\textcolor{black}{,
and} \\
 (iii) $\wt{W}$ and $N_{1},...,N_{M}$ are independent of \textcolor{black}{each
other}.

Thus, under $\mQ$, $N(t)=\sum_{m=1}^{M}N_{m}(t)$ is Poisson with
intensity $\wt{\l}=\sum_{m=1}^{M}\wt{\l}_{m}$, $Y_{i}$ are \textcolor{black}{i.i.d.}
with $\mQ(Y_{i}=y_{m})\coloneqq q(y_{m})\coloneqq\frac{\wt{\l}_{m}}{\wt{\l}}$\textcolor{black}{,
and} $Q(t)-\wt{\b}\wt{\l}t$ is a martingale, where 
$\wt{\b}\coloneqq\mE_{\mQ}Y_{i}=\sum_{m=1}^{M}y_{m}q(y_{m})=\frac{1}{\wt{\l}}\sum_{m=1}^{M}\wt{\l}_{m}y_{m}.$
\textcolor{black}{This implies that} $\mQ$ is\textcolor{black}{{} an
EMM }if and only if the discounted price process $S$ is a $\mQ$-martingale\textcolor{black}{,
i.e. }
\begin{align*}
dS(t) & =(\a-\b\l)S(t)dt+\s S(t)dW(t)+S(t-)dQ(t)\\
 & =rS(t)dt+\s S(t)d\wt{W}(t)+S(t-)d(Q(t)-\wt{\b}\wt{\l}t).
\end{align*}
Combined, these equalities imply that\textcolor{black}{{} the solutions
to the} market price of risk equation\textcolor{black}{, which is }$\a-\b\l=r+\s\th-\wt{\b}\wt{\l},$ \textcolor{black}{or, rewritten}\textcolor{black}{{} as} $\a-r=\s\th+\sum_{m=1}^{M}(\l_{m}-\wt{\l}_{m})y_{m},$
\textcolor{black}{characterize}\textcolor{black}{{} the set of EMMs in
the original market. Because }\textcolor{black}{this is one}\textcolor{black}{{}
equation with }\textcolor{black}{the} $M+1$ unknowns ($\th,\wt{\l}_{1},...,\wt{\l}_{M}$),
the market is incomplete.

\textcolor{black}{We now consider the original market as constructed
above, but with three stocks trading, denoted $S_{i}$ for $i=1,2,3$,
each following the }\textcolor{black}{similar}\textcolor{black}{{} stochastic
evolution. As long as the number of jump processes $M\geq3,$ the
original market is still arbitrage-free and incomplete. To simplify
the presentation, we assume that $M=3.$}

Including the subscripts for the three traded stocks,
the original market now consists of a single Brownian motion and three independent Poisson processes $N_{i}$ with intensities $\l_{i}$ defining three compound Poisson processes $Q_{i}(t)=y_{i,1}N_{1}+y_{i,2}N_{2}(t)+y_{i,3}N_{3}(t),\quad i=1,2,3,$ where $y_{i,m}>-1$ for $i=1,2,3$, $m=1,2,3$. Let $\b_{i}=\frac{1}{\l}(\l_{1}y_{i,1}+\l_{2}y_{i,2}+\l_{3}y_{i,3}),\quad i=1,2,3.$ The stocks follow the stochastic differential equations 
\begin{align*}
dS_{i}(t) =(\a_{i}-\b_{i}\l)S_{i}(t)dt+\s_{i}S_{i}(t)dW(t)+S_{i}(t-)dQ_{i}(t),\quad i=1,2,3.
\end{align*}
The \textcolor{black}{modified} market price of risk equations are
\begin{align}
\begin{split}\a_{1}-r & =\s_{1}\th+(\l_{1}-\wt{\l}_{1})y_{1,1}+(\l_{2}-\wt{\l}_{2})y_{1,2}+(\l_{3}-\wt{\l}_{3})y_{1,3},\\
\a_{2}-r & =\s_{2}\th+(\l_{1}-\wt{\l}_{1})y_{2,1}+(\l_{2}-\wt{\l}_{2})y_{2,2}+(\l_{3}-\wt{\l}_{3})y_{2,3},\\
\a_{3}-r & =\s_{3}\th+(\l_{1}-\wt{\l}_{1})y_{3,1}+(\l_{2}-\wt{\l}_{2})y_{3,2}+(\l_{3}-\wt{\l}_{3})y_{3,3}.
\end{split}
\label{eq: original mkt risk premiums}
\end{align}
\textcolor{black}{For subsequent use, we let $\mG$ be} the augmented
filtration generated by $W,N_{1},N_{2}$ and $N_{3}$. \textcolor{black}{We
note that this structure is a modification of an example from Shreve,
2004 \cite{key-16}, }Example 11.7.4, p.515-516.

\subsubsection{The Fictitious Market}

\textcolor{black}{The purpose of this section is to construct a complete
fictitious market from the original market using filtration reduction
from $\mG$ to an }$\mF\subseteq\mathbbm{G}$\textcolor{black}{. To
simplify the exposition, we consider the original market consisting
}\textcolor{black}{of the}\textcolor{black}{{} three stocks. The generalization
to an arbitrary number of stocks is subsequently presented. }

\textcolor{black}{Consider the reduced filtration} $\mF$ generated
by $W,N_{1}$ and $N_{2}$.\textcolor{black}{{} Denoting by $\wt{S}$
the vector of price processes in the fictitious market, we} would
like to have $\wt{S}(t)=\mE[S(t)|\cF_{t}]$ for each $t\in[0,T]$.
This equality to be understood as an optional projection. Note that
\begin{align*}
\wt{S}_{i}(t) & =\mE[S_{i}(t)|\cF_{t}]=\mE[S_{i}(0)\exp\{\s_{i}W(t)+(\a_{i}-\b_{i}\l-\frac{1}{2}\s_{i}^{2})t\}\prod_{j=1}^{N(t)}(Y_{ij}+1)|\cF_{t}]\\
 & =S_{i}(0)\exp\{\s_{i}W(t)+(\a_{i}-\b_{i}\l-\frac{1}{2}\s_{i}^{2})t\}\mE[\prod_{j=1}^{N(t)}(Y_{ij}+1)|\cF_{t}]\\
 & =S_{i}(0)\exp\{\s_{i}W(t)+(\a_{i}-\b_{i}\l-\frac{1}{2}\s_{i}^{2})t\}\mE[\prod_{m=1}^{M}(y_{i,m}+1)^{N_{m}(t)}|\cF_{t}]\\
 & =S_{i}(0)\exp\{\s_{i}W(t)+(\a_{i}-(\l_{1}y_{i,1}+\l_{2}y_{i,2}+\l_{3}y_{i,3})-\frac{1}{2}\s_{i}^{2})t\}\times\\
 & \times\mE[\prod_{m=1}^{M}(y_{i,m}+1)^{N_{m}(t)}|\cF_{t}]\\
 & =S_{i}(0)\exp\{\s_{i}W(t)+(\a_{i}-(\l_{1}y_{i,1}+\l_{2}y_{i,2}+\l_{3}y_{i,3})-\frac{1}{2}\s_{i}^{2})t\}\times\\
 & \times\prod_{m=1}^{2}(y_{i,m}+1)^{N_{m}(t)}\mE[(y_{i,3}+1)^{N_{3}(t)}|\cF_{t}]\\
 & =S_{i}(0)\exp\{\s_{i}W(t)+(\a_{i}-(\l_{1}y_{i,1}+\l_{2}y_{i,2})-\frac{1}{2}\s_{i}^{2})t\}e^{-\l_{3}y_{i,3}t}\times\\
 & \times\prod_{j=1}^{N_{1}(t)+N_{2}(t)}(\wt{Y}_{ij}+1)\mE[(y_{i,3}+1)^{N_{3}(t)}]\\
 & =S_{i}(0)\exp\{\s_{i}W(t)+(\a_{i}-\wh{\b}\wh{\l}-\frac{1}{2}\s_{i}^{2})t\}e^{-\l_{3}y_{i,3}t}\prod_{j=1}^{\wt{N}(t)}(\wt{Y}_{ij}+1)e^{\l p(y_{i,3})ty_{i,3}}\\
 & =S_{i}(0)\exp\{\s_{i}W(t)+(\a_{i}-\wh{\b}\wh{\l}-\frac{1}{2}\s_{i}^{2})t\}e^{-\l_{3}y_{i,3}t}\prod_{j=1}^{\wt{N}(t)}(\wt{Y}_{ij}+1)e^{\l_{3}ty_{i,3}}\\
 & =S_{i}(0)\exp\{\s_{i}W(t)+(\a_{i}-\wh{\b}\wh{\l}-\frac{1}{2}\s_{i}^{2})t\}\prod_{j=1}^{\wt{N}(t)}(\wt{Y}_{ij}+1),
\end{align*}
where we have used the facts that \textcolor{black}{$M=3$,} $\prod_{j=1}^{N(t)}(Y_{ij}+1)=\prod_{m=1}^{M}(y_{i,m}+1)^{N_{m}(t)}$,
$\prod_{m=1}^{2}(y_{i,m}+1)^{N_{m}(t)}$ is $\cF_{t}$-measurable,
$(y_{i,3}+1)^{N_{3}(t)}$ is independent of $\cF_{t}$, \textcolor{black}{the
notation} $\wt{N}(t)\coloneqq N_{1}(t)+N_{2}(t)$, $\wh{\b}_{i}\coloneqq\frac{1}{\wh{\l}}(\l_{1}y_{i,1}+\l_{2}y_{i,2}),\quad i=1,2,3,$
$\wh{\l}\coloneqq\l_{1}+\l_{2}$, and $\wt{Y}_{ij}$ \textcolor{black}{for
i.i.d}. random variables which \textcolor{black}{take the} values $y_{1}$
and $y_{2}$. \textcolor{black}{The derivation also used} the probability
generating \textcolor{black}{function for} $N_{3}(t)\sim$ Poisson$(\l p(y_{i,3})t)$.

The change from \textcolor{black}{the jump sizes} $Y_{ij}$ that are
$\{y_{1},y_{2},y_{3}\}$-distributed to \textcolor{black}{the jump sizes}
$\wt{Y}_{ij}$ that are $\{y_{1},y_{2}\}$-distributed \textcolor{black}{corresponds
}\textcolor{black}{to a filtration}\textcolor{black}{{} reduction with
complete neglect. Here, the trader ``neglects'' the jump size }$y_{3}$,\textcolor{black}{{}
which corresponds to a reduced filtration} generated by a subset of
$N^{1},N^{2},N^{3}$. \textcolor{black}{Clearly other} reductions are
possible\textcolor{black}{, and in the next section} we present \textcolor{black}{a
filtration} reduction \textcolor{black}{with partial} neglect.

\textcolor{black}{Hence, the s}tock prices in the fictitious market
\textcolor{black}{follow the stochastic }differential equations 
\begin{align*}
d\wt{S}_{i}(t) & =(\a_{i}-\wh{\b}_{i}\wh{\l})S_{i}(t)dt+\s_{i}\wt{S}_{i}(t)dW(t)+\wt{S}_{i}(t-)d\wt{Q}_{i}(t),\quad i=1,2,3,
\end{align*}
where $\wt{Q}_{i}(t)=y_{i,1}N_{1}(t)+y_{i,2}N_{2}(t),\quad i=1,2,3,$ with $y_{i,m}>-1$ for $i=1,2,3$, $m=1,2$. The solution is $\wt{S}_{i}(t)=S_{i}(0)\exp\{\s_{i}W(t)+(\a_{i}-\wh{\b}\wh{\l}-\frac{1}{2}\s_{i}^{2})t\}\prod_{j=1}^{\wt{N}(t)}(Y_{ij}+1).$
\textcolor{black}{The corresponding market price }of risk equations
are 
\begin{align}
\a_{1}-r & =\s_{1}\th+(\l_{1}-\wt{\l}_{1})y_{1,1}+(\l_{2}-\wt{\l}_{2})y_{1,2},\label{eq: uplifed EMM}\\
\a_{2}-r & =\s_{2}\th+(\l_{1}-\wt{\l}_{1})y_{2,1}+(\l_{2}-\wt{\l}_{2})y_{2,2},\nonumber \\
\a_{3}-r & =\s_{3}\th+(\l_{1}-\wt{\l}_{1})y_{3,1}+(\l_{2}-\wt{\l}_{2})y_{3,2}.\nonumber 
\end{align}
There are three equations with three unknowns ($\th,\wt{\l}_{1},\wt{\l}_{2}$).
If there exists a unique solution $(\th^{*},\wt{\l}_{1}^{*},\wt{\l}_{2}^{*})$
to the market price of risk equations, then the fictitious market
is arbitrage-free and complete. \textcolor{black}{By the second fundamental
theorem of asset pricing, }\textcolor{black}{this implies a}\textcolor{black}{{}
unique EMM $\wt{\mQ}$ under which $\wt{S}$ is a martingale.}

\subsubsection{\textcolor{black}{The Uplifted EMM }}

\textcolor{black}{To price derivatives in the original market, we
need to uplift the unique EMM }\textcolor{black}{from the}\textcolor{black}{{}
fictitious market to an EMM in the original market. In this simple
case,} the uplift will be unique, because given the vector $(\th^{*},\wt{\l}_{1}^{*},\wt{\l}_{2}^{*})$,
the only $\wt{\l}_{3}$ that \textcolor{black}{satisfies} the set
of equations (\ref{eq: original mkt risk premiums}) and (\ref{eq: uplifed EMM})
is $\wt{\l}_{3}=\l_{3}$. \textcolor{black}{We state this result as
a theorem. } 
\begin{thm}
\textcolor{black}{(A Unique Uplifted EMM) }\textcolor{black}{Consider
a market} with one Brownian motion, three Poisson processes, and three
stocks. Assume that in the fictitious market $\mF$ obtained by complete
neglect there is a unique EMM defined by the vector $(\th^{*},\wt{\l}_{1}^{*},\wt{\l}_{2}^{*})$.
Then, there is a unique uplift of the unique EMM from the\textcolor{black}{{}
fictitious market to the original} market \textcolor{black}{given}
by the vector $(\th^{*},\wt{\l}_{1}^{*},\wt{\l}_{2}^{*},\l_{3})$. 
\end{thm}

\textcolor{black}{It is easy to see that the restriction of the uplifted
EMM defined on $\cG_{T}$ to the $\s$-field $\cF_{T}$ is precisely
the measure $\wt{\mQ}$, i.e. } 
\begin{thm}
$\mQ^{*}|_{\cF_{T}}=\wt{\mQ}.$ 
\end{thm}

\begin{proof}
Take $A\in\cF_{T}.$ Then $\mQ(A)=\int_{A}Z_{T}d\mP=\int_{A}\mE[Z_{T}|\cF_{T}]d\mP=\int_{A}\wt{Z}_{T}d\mP=\wt{\mQ}(A).$
\end{proof}

\subsection{\textcolor{black}{$n$ Stocks, One Brownian Motion, and $M$ Poisson
Processes}}

\textcolor{black}{This section extends the previous original market
to include $n$ traded stocks and $M$ Poisson processes.}

\subsubsection{The Original Market}

Assume that there are $n$ stock price processes driven by a single
Brownian motion and $M$ independent Poisson processes with intensities
$\l_{m},m=1,...,M$, \textcolor{black}{which are} independent of the
Brownian motion. We assume that $n<1+M$ and the filtration $\mG$
is generated by the $1+M$ processes. Thus, we have 
\begin{align*}
dS_{i}(t) & =(\a_{i}-\b_{i}\l)S(t)dt+\s_{i}S(t)dW(t)+S_{i}(t-)dQ_{i}(t),\quad i=1,...,n,\\
Q_{i}(t) & =y_{i,1}N_{1}+y_{i,2}N_{2}(t)+...+y_{i,M}N_{M}(t),\quad i=1,...,n,
\end{align*}
where $y_{i,m}>-1$ for $i=1,...,n$, $m=1,...,M$, $\b_{i}=\frac{1}{\l}(\l_{1}y_{i,1}+\l_{2}y_{i,2}+...+\l_{M}y_{i,M}),\quad i=1,...,n$\textcolor{black}{,
and }$\l=\sum_{m=1}^{M}\l_{m}$.

The corresponding market price of risk equations are given by 
\begin{align*}
\a_{1}-r & =\s_{1}\th+(\l_{1}-\wt{\l}_{1})y_{1,1}+...+(\l_{M}-\wt{\l}_{M})y_{1,M},\\
...\\
\a_{n}-r & =\s_{n}\th+(\l_{1}-\wt{\l}_{1})y_{n,1}+...+(\l_{M}-\wt{\l}_{M})y_{n,M}.
\end{align*}
This market is incomplete because we have $n$ equations \textcolor{black}{with
the} $1+M$ unknowns $\th,\wt{\l}_{1},...,\wt{\l}_{M}$, where $n<1+M$.

\subsubsection{The Fictitious Market}

\textcolor{black}{This section constructs the complete fictitious market.
}\textcolor{black}{To obtain} a complete market, we need a reduction
to $n$ sources of randomness.\textcolor{black}{{} We }\textcolor{black}{choose
to reduce to the filtration that is} generated by the single Brownian
motion and $n-1<M$ Poisson processes. Without loss of generality,
we retain the first $n-1$ Poisson processes and discard $N_{n},N_{n+1},...,N_{M}$.

Thus, \textcolor{black}{the stock price evolution in the fictitious
market is} 
\begin{align*}
d\wt{S}_{i}(t) & =(\a_{i}-\wh{\b}_{i}\wh{\l})S(t)dt+\s_{i}S(t)dW(t)+S_{i}(t-)d\wt{Q}_{i}(t),\quad i=1,...,n,\\
\wt{Q}_{i}(t) & =y_{i,1}N_{1}(t)+y_{i,2}N_{2}(t)+...+y_{i,n-1}N_{n-1}(t),\quad i=1,...,n,
\end{align*}
where $y_{i,m}>-1$ for $i=1,...,n$, $m=1,...,n-1$, $\wh{\b}_{i}=\frac{1}{\wh{\l}}(\l_{1}y_{i,1}+\l_{2}y_{i,2}+...+\l_{n-1}y_{i,n-1}),\quad i=1,...,n$,
and $\wh{\l}=\sum_{m=1}^{n-1}\l_{m}$.

The corresponding market price of risk equations are 
\begin{align*}
\a_{1}-r & =\s_{1}\th+(\l_{1}-\wt{\l}_{1})y_{1,1}+...+(\l_{n-1}-\wt{\l}_{n-1})y_{1,n-1},\\
...\\
\a_{n}-r & =\s_{n}\th+(\l_{1}-\wt{\l}_{1})y_{3,1}+...+(\l_{n-1}-\wt{\l}_{n-1})y_{n,n-1}.
\end{align*}
\textcolor{black}{This }\textcolor{black}{gives }$n$ equations in $n$
unknowns $\th,\wt{\l}_{1},...,\wt{\l}_{n-1}$. Assuming that it has
a unique solution, we obtain an arbitrage-free and complete market.
\textcolor{black}{Here the}\textcolor{black}{{} unique EMM $\wt{\mQ}$
is} given by $\th^{*},\wt{\l}_{1}^{*},...,\wt{\l}_{n-1}^{*}$.

\subsubsection{\textcolor{black}{The Uplifted EMM}: Issues of Uniqueness }

As before, we now \textcolor{black}{uplift the EMM }\textcolor{black}{from
the}\textcolor{black}{{} fictitious market to a unique EMM in the original
market.} \textcolor{black}{Unfortunately, uniqueness fails} \textcolor{black}{because}
if $M-(n-1)>1$,\textcolor{black}{{} then there} are (infinitely) many
ways in which the original market price of risk equations can be satisfied,
\textcolor{black}{given }\textcolor{black}{that the} values $\th^{*},\wt{\l}_{1}^{*},...,\wt{\l}_{n-1}^{*}$
\textcolor{black}{are already fixed} \textcolor{black}{by the EMM in
the fictitious market.} More precisely, any $\wt{\l}_{n},...,\wt{\l}_{M}$
that satisfy 
\begin{align*}
(\l_{n}-\wt{\l}_{n})y_{1,n}+...+(\l_{M}-\wt{\l}_{M})y_{1,M} & =0,\\
...\\
(\l_{n}-\wt{\l}_{n})y_{n,n}+...+(\l_{M}-\wt{\l}_{M})y_{n,M} & =0
\end{align*}
will be viable \textcolor{black}{candidates for an uplifted EMM. To
obtain a unique EMM in the original market, we apply the notion of
}\textcolor{black}{consistency from Grigorian and Jarrow 2024}\textcolor{black}{{}
\cite{key-12,key-13,key-14}. }

\textcolor{black}{To see how}\textcolor{black}{{} consistency manifests
itself in a jump-diffusion }\textcolor{black}{setting, suppose that
}\textcolor{black}{we selected a }different subset of \textcolor{black}{the
}$M=n+1$\textcolor{black}{{} Poisson} \textcolor{black}{processes available
for}\textcolor{black}{{} the filtration reduction.} \textcolor{black}{Assume
that }we retain $N_{1},...,N_{n-2},N_{n}$, i.e. we choose $N_{n}$
instead of $N_{n-1}$. \textcolor{black}{As before, we still} discard
$N_{n+1}$ in both scenarios. Note that the \textcolor{black}{different
risk assessments} disagree on the importance of $N_{n-1}$ and $N_{n}$,
but both agree that $N_{n+1}$'s risk is irrelevant. Therefore, it
is natural to expect that \textcolor{black}{the uplifts in both filtration
reductions} would agree \textcolor{black}{on the subset of risks considered irrelevant
}\textcolor{black}{and} \textcolor{black}{the corresponding intensity} $\wt{\l}_{n+1}$
\textcolor{black}{would be the same.}

More formally, consistency requires that the uplifts \textcolor{black}{agree
on} the overlapping views of irrelevant risks. \textcolor{black}{The
only uplift} with this consistency property is the one obtained by
setting $\wt{\l}_{m}=\l_{m}$ for all $m$ not included in the reduced
filtration. \textcolor{black}{To see this, consider two reductions,
$A$ and $B$, whereby $A$ involves discarding $N_{n}$ and $N_{n+1}$,
while $B$ involves discarding $N_{n-1}$ and $N_{n+1}$. The non-uniqueness
stems from the following systems of equations for the market prices
of risk} 
\begin{align*}
\begin{split}\text{Reduction}\quad A\quad\\
(\l_{n}-\wt{\l}_{n})y_{1,n}+(\l_{n+1}-\wt{\l}_{n+1})y_{1,n+1} & =0,\\
...\\
(\l_{n}-\wt{\l}_{n})y_{n,n}+(\l_{n+1}-\wt{\l}_{n+1})y_{n,n+1} & =0,
\end{split}
\begin{split}\text{Reduction}\quad B\quad\\
(\l_{n-1}-\wt{\l}_{n-1})y_{1,n-1}+(\l_{n+1}-\wt{\l}_{n+1})y_{1,n+1} & =0,\\
...\\
(\l_{n-1}-\wt{\l}_{n-1})y_{n,n-1}+(\l_{n+1}-\wt{\l}_{n+1})y_{n,n+1} & =0.
\end{split}
\end{align*}
\textcolor{black}{Let us assume that} 
\begin{align*}
(\l_{n+1}-\wt{\l}_{n+1})y_{1,n+1} & \coloneqq x_{1}\neq0,\\
...\\
(\l_{n+1}-\wt{\l}_{n+1})y_{n,n+1} & \coloneqq x_{n}\neq0.
\end{align*}
\textcolor{black}{Then, consistency between $A$ and $B$ requires
that} 
\begin{align*}
(\l_{n-1}-\wt{\l}_{n-1})y_{1,n-1}=(\l_{n}-\wt{\l}_{n})y_{1,n} & =-x_{1},\\
...\\
(\l_{n-1}-\wt{\l}_{n-1})y_{n,n-1}=(\l_{n}-\wt{\l}_{n})y_{n,n} & =-x_{n}.
\end{align*}
\textcolor{black}{Next, consider a third reduction $C$, where }\textcolor{black}{we
neglect}\textcolor{black}{{} $N_{n-1}$ and $N_{n}$. The corresponding
market price of risk equations are} 
\begin{align*}
\text{Reduction}\quad C\quad\quad\quad\\
(\l_{n-1}-\wt{\l}_{n-1})y_{1,n-1}+(\l_{n}-\wt{\l}_{n})y_{1,n} & =0,\\
...\\
(\l_{n-1}-\wt{\l}_{n-1})y_{n,n-1}+(\l_{n}-\wt{\l}_{n})y_{n,n} & =0.
\end{align*}
\textcolor{black}{If we require that all the three choices are consistent,
substitution of the previously obtained values yield}s 
\begin{align*}
-2x_{1} & =0,\\
...\\
-2x_{n} & =0,
\end{align*}
and hence a contradiction. Therefore $x_{m}=0$, $m=1,...,n,$ and,
since the $y$'s are nonzero, $\wt{\l}_{n+1}=\l_{n+1}$. Arguing similarly,
we \textcolor{black}{see that consistency requires} that $\wt{\l}_{m}=\l_{m}$
for all $m$ in the subset of discarded processes.

A somewhat different argument leads to the same conclusion. Let us
still assume that there are $M=n+1$ Poisson processes, and we intend
to \textcolor{black}{remove two} Poisson processes, say $N_{n}$ and
$N_{n+1}$ to obtain a complete market and a unique EMM $\wt{\mQ}$
\textcolor{black}{in the fictitious market.} Without the notion of
consistency, the uplift is not unique, as is evident from the market
prices of risk equations 
\begin{align*}
(\l_{n}-\wt{\l}_{n})y_{1,n}+(\l_{n+1}-\wt{\l}_{n+1})y_{1,n+1} & =0,\\
...\\
(\l_{n}-\wt{\l}_{n})y_{n,n}+(\l_{n+1}-\wt{\l}_{n+1})y_{n,n+1} & =0.
\end{align*}
However, if instead of uplifting to the original market, we uplift
to a market with more information stemming from a single additional
Poisson process (say, $N_{n}$), uniqueness can be recovered. Indeed,
in this case we will first obtain $\th^{*},\wt{\l}_{1}^{*},...,\wt{\l}_{n-1}^{*}$
and uplift it to $\mF^{(+1)}$ generated by $\mF$ and $N_{n}$. This
corresponds to solving 
\begin{align*}
(\l_{n}-\wt{\l}_{n})y_{1,n} & =0,\\
...\\
(\l_{n}-\wt{\l}_{n})y_{n,n} & =0,
\end{align*}
from which we obtain $\wt{\l}_{n}^{*}=\l_{n}$. Next, we uplift again
to the original $\mG=\mF^{(+2)}$ and obtain $\wt{\l}_{n+1}^{*}=\l_{n+1}$.
Clearly, this \textbf{sequential procedure} can be applied \textcolor{black}{to
an }arbitrary number of discrete jumps, \textcolor{black}{yielding}
$\wt{\l}_{m}^{*}=\l_{m}$ for all the $m$ in the subset of\textcolor{black}{{}
the neglected }Poisson processes.

We have thus proved the following theorem. 
\begin{thm}
\textcolor{black}{(A Unique Consistent Uplifted EMM) }Assume the reduced
market obtained by complete neglect is arbitrage-free and complete.
Then, there exists a unique consistent uplifted EMM $\mQ^{*}$ of
the unique EMM $\wt{\mQ}$ in the fictitious market given by $\th^{*},\wt{\l}_{1}^{*},...,\wt{\l}_{n-1}^{*}$,
and $\wt{\l}_{m}^{*}=\l_{m},m=n,...,M$. 
\end{thm}

\subsection{\textcolor{black}{$n$ Stocks, $D$ Brownian Motions, and $M$ Poisson
Processes}}

\textcolor{black}{Finally, this section extends the previous insights
to the general case.}

\subsubsection{The Original Market}

Let there are $n$ stock price processes driven by $D$ independent
Brownian motions and $M$ independent Poisson processes with intensities
$\l_{m},m=1,...,M$, which are independent of the Brownian \textcolor{black}{motions.
}We assume that $n<D+M$ and the filtration $\mG$ is generated by
the $D+M$ processes. \textcolor{black}{The stock price process is}
\begin{align*}
dS_{i}(t) & =(\a_{i}-\b_{i}\l)S(t)dt+\sum_{j=1}^{D}\s_{ij}S(t)dW^{j}(t)+S_{i}(t-)dQ_{i}(t),\quad i=1,...,n,\\
Q_{i}(t) & =y_{i,1}N_{1}(t)+y_{i,2}N_{2}(t)+...+y_{i,M}N_{M}(t),\quad i=1,...,n,
\end{align*}
where $y_{i,m}>-1$ for $i=1,...,n$, $m=1,...,M$, $\b_{i}=\frac{1}{\l}(\l_{1}y_{i,1}+\l_{2}y_{i,2}+...+\l_{M}y_{i,M}),\quad i=1,...,n$
and $\l=\sum_{m=1}^{M}\l_{m}$.

The corresponding market price of risk equations are 
\begin{align*}
\a_{1}-r & =\s_{11}\th_{1}+...+\s_{1D}\th_{D}+(\l_{1}-\wt{\l}_{1})y_{1,1}+...+(\l_{M}-\wt{\l}_{M})y_{1,M},\\
...\\
\a_{n}-r & =\s_{n1}\th_{1}+...+\s_{nD}\th_{D}+(\l_{1}-\wt{\l}_{1})y_{n,1}+...+(\l_{M}-\wt{\l}_{M})y_{n,M}.
\end{align*}
The market is arbitrage-free \textcolor{black}{if the equations have
a solution and }incomplete, \textcolor{black}{because there are} $n$
equations \textcolor{black}{with the} $D+M$ unknowns $\th_{1},...,\th_{D},\wt{\l}_{1},...,\wt{\l}_{M}$.

\subsubsection{The Fictitious Market}

To obtain a complete market, we reduce the filtration to one generated
by $n$ sources of randomness. Obviously, this can be done in many
different ways. We consider several examples \textcolor{black}{from
which a pattern }\textcolor{black}{emerges.}

{ \centering \textbf{Case 1: Reduction to a Brownian Market} \\
 } \bigskip{}

Assume $n<M$ and $n<D$. We choose \textcolor{black}{the reduction}
$\mF$ generated\textcolor{black}{{} by the} $n$ Brownian motions, i.e.
we do not include any Poisson processes and reduce the number of Brownian
motions from $D$ to $n$. \textcolor{black}{Thus, the price process
in the fictitious market is} 
\begin{align*}
d\wt{S}_{i}(t) & =\wt{S}_{i}(t)\left((\a_{i}-r)dt+\sum_{d=1}^{n}\s_{id}dW^{d}(t)\right),\\
\wt{S}_{i}(t) & =S_{i}(0)e^{(\a_{i}-r)t-\frac{1}{2}\left(\sum_{d=1}^{D}\s_{id}^{2}\right)t+\sum_{d=1}^{D}\s_{id}W^{d}(t)}
\end{align*}
for all $t\in[0,T]$ and $i=1,...,n$. It is clear that this i\textcolor{black}{s
a Brownian motion market with constant coefficients. Hence,} \textcolor{black}{this
is }a complete market and, by the second fundamental theorem of asset
pricing, the unique EMM $\wt{\mQ}$ is given by $\wt{\th}=\wt{\s}^{-1}(\a-r\wt{\1})$,
where $\wt{\s}$ refers to the original volatility matrix\textcolor{black}{{}
truncated to} the corresponding \textcolor{black}{relevant }$n$ terms.

\bigskip{}

{ \centering \textbf{The Uplifted EMM} \\
 } \bigskip{}

From the \textcolor{black}{previous section}, if the reduced market
is arbitrage-free and complete, there exists a \textcolor{black}{unique
consistent} uplift $\mQ^{*}$ of the unique EMM $\wt{\mQ}$ \textcolor{black}{from}\textcolor{black}{{}
the fictitious market} to the original market given by $\th^{*}\coloneqq(\wt{\th},\mathbf{0}_{D-n})$.

\textcolor{black}{Combining this with the previous analysis, we can
obtain }\textcolor{black}{a unique consistent uplift}\textcolor{black}{{}
to the original jump-diffusion market with discrete jumps.} 
\begin{thm}
\textcolor{black}{(A Unique Consistent Uplifted EMM)} Assume the reduced
market obtained by complete neglect is arbitrage-free and complete.
Then, there exists a unique consistent uplift $\mQ^{*}$ of the unique
EMM $\wt{\mQ}$ \textcolor{black}{from}\textcolor{black}{{} the fictitious
market }to the original market given by $\th^{*}$ and $\wt{\l}_{m}=\l_{m},m=1,...,M$. 
\end{thm}

\bigskip{}

{ \centering \textbf{Case 2: Reduction to a Jump-Diffusion Market}
\\
 } \bigskip{}

Assume $n<M$ and $n\geq D$. \textcolor{black}{Here, we }choose \textcolor{black}{a}
reduction $\mF$ generated by the full set of $D$ Brownian motions
and $n-D$ Poisson processes. \textcolor{black}{Thus, the price process
in the fictitious market is} 
\begin{align*}
d\wt{S}_{i}(t) & =(\a_{i}-\wh{\b}_{i}\wh{\l})S(t)dt+\sum_{j=1}^{D}\s_{ij}S(t)dW^{j}(t)+S_{i}(t-)d\wt{Q}_{i}(t),\quad i=1,...,n,\\
\wt{Q}_{i}(t) & =y_{i,1}N_{1}(t)+y_{i,2}N_{2}(t)+...+y_{i,n-D}N_{n-D}(t),\quad i=1,...,n,
\end{align*}
where $y_{i,m}>-1$ for $i=1,...,n$, $m=1,...,n-D$, $\wh{\b}_{i}=\frac{1}{\wh{\l}}(\l_{1}y_{i,1}+\l_{2}y_{i,2}+...+\l_{n-D}y_{i,n-D}),\quad i=1,...,n$,
and $\wh{\l}=\sum_{m=1}^{n-D}\l_{m}$.

The corresponding market price of risk equations are 
\begin{align*}
\a_{1}-r & =\s_{11}\th_{1}+...+\s_{1D}\th_{D}+(\l_{1}-\wt{\l}_{1})y_{1,1}+...+(\l_{n-D}-\wt{\l}_{n-D})y_{1,n-D},\\
...\\
\a_{n}-r & =\s_{n1}\th_{1}+...+\s_{nD}\th_{D}+(\l_{1}-\wt{\l}_{1})y_{3,1}+...+(\l_{n-D}-\wt{\l}_{n-D})y_{n,n-D}.
\end{align*}
\textcolor{black}{There are $n$ equations in $n$ unknowns} $\th_{1},...,\th_{D},\wt{\l}_{1},...,\wt{\l}_{n-D}$.
Assuming that it has a unique solution, we obtain an arbitrage-free
and complete market, and hence a unique EMM $\wt{\mQ}$ given by $\wt{\th}^{*}$
and $\wt{\l}_{1}^{*},...,\wt{\l}_{n-D}^{*}$.

\bigskip{}

{ \centering \textbf{The Uplifted EMM} \\
 } \bigskip{}

Following the same analysis as above, we uplift this measure to the
original jump-diffusion \textcolor{black}{market.}
\begin{thm}
\textcolor{black}{(A Unique Consistent Uplifted EMM)} Assume the reduced
market obtained by complete neglect is arbitrage-free and complete.
Then, there exists a unique consistent uplifted EMM $\mQ^{*}$ \textcolor{black}{from
the} unique EMM $\wt{\mQ}$ \textcolor{black}{in the fictitious market}
given by $\wt{\th}_{1}^{*},...,\wt{\th}_{D}^{*}$, $\wt{\l}_{1}^{*},...,\wt{\l}_{n-D}^{*},$
and $\wt{\l}_{m}=\l_{m},m=n-D+1,...,M$. 
\end{thm}

\subsection{Time-Varying Coefficients and Inhomogeneous Poisson Processes}

\textcolor{black}{We now extend the }\textcolor{black}{analytics to}\textcolor{black}{{}
Poisson processes with intensities that are deterministic functions
of time. Also, the mean rate of return $\a(t)$, the money market
rate $r(t)$, and the volatility $\s(t)$ can be deterministic functions
of time as well, with $\s(t)$ being square integrable}.

\subsubsection{\textcolor{black}{Preliminaries}}

\textcolor{black}{To understand }\textcolor{black}{the changes necessary}\textcolor{black}{,
we recast the constant coefficient case to the more general notation
of stochastic integrals with respect to random measures and their
compensators. Note 
\begin{align*}
\prod_{i=1}^{N(t)}(1+Y_{i})=e^{\ln\prod_{i=1}^{N(t)}(1+Y_{i})}=e^{\sum_{i=1}^{N(t)}\ln(1+Y_{i})}=e^{\int_{0}^{t}\int_{\mR}\ln(1+y)N(ds,dy)},
\end{align*}
and hence the associated exponential martingale is given by 
\begin{align*}
e^{\int_{0}^{t}\int_{\mR}\ln(1+y)N(ds,dy)-\int_{0}^{t}\int_{\mR}y\nu(ds,dy)} & =e^{\int_{0}^{t}\int_{\mR}\ln(1+y)N(ds,dy)-\int_{0}^{t}\int_{\mR}yds\l(dy)}\\
 & =e^{\int_{0}^{t}\int_{\mR}\ln(1+y)N(ds,dy)-\l t\int_{\mR}y\G(dy)}\\
 & =e^{\int_{0}^{t}\int_{\mR}\ln(1+y)N(ds,dy)-\l t\b},
\end{align*}
where we have used the fact }\textcolor{black}{that for a}\textcolor{black}{{}
compound Poisson (and hence L\'evy) process, the compensator }\textcolor{black}{factors}\textcolor{black}{{}
as $\nu(ds,dy)=ds\l(dy)$, with $\l$ the L\'evy measure of jumps, $\G(dy)=\frac{\l(dy)}{\l(\mR)}$
the distribution of $Y_{i}$, $\l\coloneqq\l(\mR)$, and $\b\coloneqq\int_{\mR}y\G(dy)$. }

\textcolor{black}{When the intensity $\l$ of the Poisson process
is not}\textcolor{black}{{} a constant}\textcolor{black}{, the resulting
marked point process in the exponent is no longer a L\'evy process,
since it does not have stationary increments. }\textcolor{black}{Specifically,
the}\textcolor{black}{{} compensator measure no longer }\textcolor{black}{factors
into}\textcolor{black}{{} a product measure, and the distribution of
the marks explicitly depends on $t$, i.e. $\nu(ds,dy)=\l_{s}(dy)ds=\l_{s}(\mR)\G_{s}(dy)ds$.
Here, the associated exponential martingale is given by 
\begin{align*}
e^{\int_{0}^{t}\int_{\mR}\ln(1+y)N(ds,dy)-\int_{0}^{t}\int_{\mR}y\nu(ds,dy)} & =e^{\int_{0}^{t}\int_{\mR}\ln(1+y)N(ds,dy)-\int_{0}^{t}\int_{\mR}y\l_{s}(dy)ds}\\
 & =e^{\int_{0}^{t}\int_{\mR}\ln(1+y)N(ds,dy)-\int_{0}^{t}\int_{\mR}y\G_{s}(dy)\l_{s}(\mR)ds}.
\end{align*}
}

\textcolor{black}{In }\textcolor{black}{this section, }\textcolor{black}{we
assume that the measures $\{\G_{t}\}$ are (uniformly in $t$) finitely
supported. For a finite time horizon $[0,T]$ this means that their
atoms can be represented as a finite set $B=\{y_{1},...,y_{M}\}$.
With a slight abuse}\textcolor{black}{{} of terminology, }\textcolor{black}{we
call this set their (common) support. Note the two cases: 
\begin{align*}
\text{L\'evy case:}\quad & \int_{\mR}yF(dy)=y_{1}p_{1}+...+y_{M}p_{M},\\
 & F(dy)=\d_{y_{1}}(dy)\frac{\l(\{y_{1}\})}{\l(B)}+...+\d_{y_{M}}(dy)\frac{\l(\{y_{M}\})}{\l(B)},\\
\text{Time-varying case:}\quad & \int_{\mR}yF_{s}(dy)=y_{1}p_{1}(s)+...+y_{M}p_{M}(s),\\
 & F_{s}(dy)=\d_{y_{1}}(dy)\frac{\l_{s}(\{y_{1}\})}{\l_{s}(B)}+...+\d_{y_{M}}(dy)\frac{\l_{s}(\{y_{M}\})}{\l_{s}(B)}.
\end{align*}
}\textcolor{black}{Within this setup, the following lemma will prove
useful.}
\begin{lem}
\textcolor{black}{\label{exp mart} Let the original price process
be  
\begin{align*}
S(t)\coloneqq e^{\int_{0}^{t}\int_{B}\ln(1+y)N(ds,dy)-\int_{0}^{t}\int_{B}y\l_{s}(dy)ds},
\end{align*}
where $N$ is a Poisson point process with finitely supported $\{\G_{t}\}$.
Then, $\mE[S_{t}]=1$ for all $t\in[0,T]$. } 
\end{lem}

\begin{proof}
\textcolor{black}{See, for example, Br\'emaud, 2021 \cite{Bremaud 2021},
p. 81, Theorem 3.2.2.} 
\end{proof}
\textcolor{black}{Although the finite support condition can be relaxed,
the above result }\textcolor{black}{suffices for }\textcolor{black}{our
purposes. Because reducing the filtration in the complete neglect
case corresponds}\textcolor{black}{{} to retaining a subset of $B$ that
will be hedged}\textcolor{black}{, we perform the following natural
decomposition of the point process corresponding to a partition of
$B$. }

\textcolor{black}{Let $\Psi$ be a marked point process with a deterministic
intensity function $\l$. }\textcolor{black}{We assume }\textcolor{black}{throughout
that $\Psi([0,t]\times B)<\infty$ for all $t\in[0,T]$. Proposition
\ref{determ is poisson} implies }\textcolor{black}{that the }\textcolor{black}{associated
counting process $N^{B}$ is a Poisson process with intensity $\l_{t}(B).$
Partition $B$ into $\wh{B}\cup\wh{B}^{c}$. By the same proposition
\ref{determ is poisson}, the associated counting processes $N^{\wh{B}}$
and $N^{\wh{B}^{c}}$ are independent Poisson processes with intensities
$\l_{t}(\wh{B})$ and $\l_{t}(\wh{B}^{c})$, respectively. Since they
are independent, $N^{\wh{B}}$ and $N^{\wh{B}^{c}}$ have no common
jumps. Hence, we obtain the decomposition of the random measure $\Psi$
as in proposition \ref{no common jumps}, i.e. 
\begin{align*}
\Psi(dt,dy)=\Psi^{\wh{B}}(dt,dy)+\Psi^{\wh{B}^{c}}(dt,dy),
\end{align*}
where $\Psi^{\wh{B}}(dt,dy)$ and $\Psi^{\wh{B}^{c}}(dt,dy)$ are
marked point processes with intensity functions $\l_{t}(\wh{B})$
and $\l_{t}(\wh{B}^{c})$, respectively with $\l_{t}(B)=\l_{t}(\wh{B}^{c})+\l_{t}(\wh{B}^{c})$.
This decomposition yields the following factorization of the original
price process 
\begin{align*}
S(t) & \coloneqq e^{\int_{0}^{t}\int_{B}\ln(1+y)\Psi(ds,dy)-\int_{0}^{t}\int_{B}y\l_{s}(dy)ds}\\
 & =e^{\int_{0}^{t}\int_{\wh{B}}\ln(1+y)\Psi^{\wh{B}}(ds,dy)+\int_{0}^{t}\int_{\wh{B}^{c}}\ln(1+y)\Psi^{\wh{B}^{c}}(ds,dy)-\int_{0}^{t}\int_{^{\wh{B}}}y\l_{s}(dy)ds-\int_{0}^{t}\int_{^{\wh{B}^{c}}}y\l_{s}(dy)ds}\\
 & =e^{\int_{0}^{t}\int_{\wh{B}}\ln(1+y)\Psi^{\wh{B}}(ds,dy)-\int_{0}^{t}\int_{^{\wh{B}}}y\l_{s}(dy)ds}e^{\int_{0}^{t}\int_{\wh{B}^{c}}\ln(1+y)\Psi^{\wh{B}^{c}}(ds,dy)-\int_{0}^{t}\int_{^{\wh{B}^{c}}}y\l_{s}(dy)ds}\\
 & =S^{\wh{B}}(t)S^{\wh{B}^{c}}(t),
\end{align*}
where $S^{\wh{B}}(t)\coloneqq e^{\int_{0}^{t}\int_{\wh{B}}\ln(1+y)\Psi^{\wh{B}}(ds,dy)-\int_{0}^{t}\int_{^{\wh{B}}}y\l_{s}(dy)ds}$
and }\\
 \textcolor{black}{$S^{\wh{B}^{c}}(t)\coloneqq e^{\int_{0}^{t}\int_{\wh{B}^{c}}\ln(1+y)\Psi^{\wh{B}^{c}}(ds,dy)-\int_{0}^{t}\int_{^{\wh{B}^{c}}}y\l_{s}(dy)ds}$. }

\textcolor{black}{Let $\mG$ be the filtration generated by $\Psi$,
and $\mF$ }\textcolor{black}{the filtration}\textcolor{black}{{} generated
by $\Psi^{\wh{B}}$. With this setup, we have the following result.
} 
\begin{thm}
\textcolor{black}{Let the original price process be 
\begin{align*}
S(t) & \coloneqq e^{\int_{0}^{t}\int_{B}\ln(1+y)\Psi(ds,dy)-\int_{0}^{t}\int_{B}y\l_{s}(dy)ds}.
\end{align*}
Then, 
\begin{align*}
\wt{S}(t)\coloneqq\mE_{\mP}[S(t)|\cF_{t}]=S^{\wh{B}}(t)=e^{\int_{0}^{t}\int_{\wh{B}}\ln(1+y)\Psi^{\wh{B}}(ds,dy)-\int_{0}^{t}\int_{^{\wh{B}}}y\l_{s}(dy)ds},
\end{align*}
where $\mE_{\mP}[S(t)|\cF_{t}]$ is to be understood in the sense
of optional projections. } 
\end{thm}

\begin{proof}
$\wt{S}(t) \coloneqq\mE_{\mP}[S(t)|\cF_{t}]=S^{\wh{B}}(t)\mE_{\mP}[S^{\wh{B}^{c}}(t)|\cF_{t}]=S^{\wh{B}}(t)\mE_{\mP}[S^{\wh{B}^{c}}(t)]=S^{\wh{B}}(t)$ by independence of $\Psi^{\wh{B}}$ and $\Psi^{\wh{B}^{c}}$ and Lemma \ref{exp mart}. 
\end{proof}
\textcolor{black}{Thus, $(S^{\wh{B}}(t),t\in\mT)$ is the fictitious
price process corresponding to the reduced filtration $\mF$}\textcolor{black}{,
reflecting those risks that the trader desires to exactly hedge.}\textcolor{black}{{}
This fact will be used freely below without explicit notification. }

\subsubsection{The Original Market}

\textcolor{black}{This section constructs the original market. }We assume
there are $n$ stocks, one Brownian motion, and $M$ Poisson processes
with deterministic intensity functions $\l_{1}(t)$,...,$\l_{M}(t)$,
where $n<M$. The restriction to one Brownian motion is to simplify
the presentation and is without loss of generality. The filtration
$\mG$ is the augmented filtration generated by the Brownian motion
and Poisson processes. \textcolor{black}{The }\textcolor{black}{original
stock}\textcolor{black}{{} price process is given by} 
\begin{align*}
dS_{i}(t)=\left(\a_{i}(t)+\int_{B}y\l_{t}(dy)\right)S_{i}(t)dt+S_{i}(t)\s_{i}(t)dW(t)+S_{i}(t-)\int_{B}y\wt{\Psi}_{i}(dy,dt),\quad i=1,...,n,
\end{align*}
where $B\subset(-1,\infty)$ is the common support of \textcolor{black}{the
jump distributions} (i.e. the mark space $\{y_{1},...,y_{M}\}$),
$\l_{t}(\cdot)$ is the measure on the mark space defined via $\l_{t}(dy)=\l_{t}(B)\G_{t}(dy)$,
with $\G_{t}(dy)$ being the distribution of the jump size given an
event at time $t$, and 
\begin{align*}
\wt{\Psi}_{i}(dy,dt)=\Psi_{i}(dy,dt)-\l_{t}(dy)dt,
\end{align*}
\textcolor{black}{where $\Psi_{i}(dy,dt)$ is the random} measure
that governs the jumps. Note that $\int_{0}^{t}\int_{B}y\wt{\Psi}_{i}(dy,ds)$
is a (local) martingale \textcolor{black}{that represents} the compensated
sum of jumps. \textcolor{black}{Note also that }$\int_{0}^{t}\int_{B}y\Psi_{i}(dy,ds)=\sum_{j=1}^{N(t)}Y_{ij}$,
\textcolor{black}{and since} the distributions are finitely supported,
$\int_{B}y\l_{t}(dy)=\sum_{m=1}^{M}y_{m}\l_{m}(t)$\textcolor{black}{.
We keep} the \textcolor{black}{integral notation} to\textcolor{black}{{}
facilitate further }extensions.

The corresponding market price of risk equations are\footnote{\textcolor{black}{Note that the expressions for jumps }\textcolor{black}{are
integrals}\textcolor{black}{{} with respect to discrete measures $\l_{t}(dy)$
and $\wt{\l}_{t}(dy).$ We }\textcolor{black}{purposely adhere}\textcolor{black}{{}
to}\textcolor{black}{{} notation similar}\textcolor{black}{{} to the previous
settings. Here, $\l_{j}(t)$ corresponds to the mass $\l_{t}(\{y_{i,j}\}),i=1,...,n,j=1,...,M$.
Similarly for $\wt{\l}_{j}(t).$}} 
\begin{align*}
\a_{1}(t)-r(t) & =\s_{1}(t)\th(t)+(\l_{1}(t)-\wt{\l}_{1}(t))y_{1,1}+...+(\l_{M}(t)-\wt{\l}_{M}(t))y_{1,M},\\
...\\
\a_{n}(t)-r(t) & =\s_{n}(t)\th(t)+(\l_{1}(t)-\wt{\l}_{1}(t))y_{n,1}+...+(\l_{M}(t)-\wt{\l}_{M}(t))y_{n,M},
\end{align*}
and the market is clearly incomplete.

\subsubsection{The Fictitious Market}

\textcolor{black}{To obtain the fictitious market, we} reduce the filtration\textcolor{black}{{}
to the one} generated by the Brownian motion and the Poisson processes
$N_{1},...,N_{n-1}$. The price process in the fictitious market satisfies:
\begin{align*}
d\wt{S}_{i}(t)=\left(\a_{i}(t)+\int_{\wh{B}}y\l_{t}(dy)\right)\wt{S}_{i}(t)dt+\wt{S}_{i}(t)\s_{i}(t)dW(t)+\wt{S}_{i}(t-)\int_{\wh{B}}y\wt{\P}_{i}(dy,dt),\quad i=1,...,n,
\end{align*}
where $\wt{\P}_{i}(dy,dt)$ is the compensated random measure $\wt{\P}_{i}(dy,dt)=\P_{i}(dy,dt)-\wh{\l}_{t}(dy)dt,$ $\int_{0}^{t}\int_{\wh{B}}y\P_{i}(dy,dt)=\sum_{j=1}^{\wt{N}(t)}\wt{Y}_{ij}$,
$\wt{Y}_{ij}$ \textcolor{black}{are} supported on $\wh{B}\coloneqq\{y_{1},...,y_{n-1}\}$,
$\wt{N}(t)=N_{1}(t)+...+N_{n-1}(t)$, and $\wh{\l}_{t}(\cdot)$ is
the restriction of the measure $\l_{t}(\cdot)$ to $\wh{B}$.

The corresponding market price of risk equations are 
\begin{align*}
\a_{1}(t)-r(t) & =\s_{1}(t)\th(t)+(\l_{1}(t)-\wt{\l}_{1}(t))y_{1,1}+...+(\l_{n-1}(t)-\wt{\l}_{n-1}(t))y_{1,n-1},\\
...\\
\a_{n}(t)-r(t) & =\s_{n}(t)\th(t)+(\l_{1}(t)-\wt{\l}_{1}(t))y_{n,1}+...+(\l_{n-1}(t)-\wt{\l}_{n-1}(t))y_{n,n-1}.
\end{align*}
\textcolor{black}{There are} $n$ equations in $n$ unknowns $\th(t),\wt{\l}_{1}(t),...,\wt{\l}_{n-1}(t)$.
If there is a unique solution, we obtain an arbitrage-free and complete
market. \textcolor{black}{This implies that there}\textcolor{black}{{}
is }a unique \textcolor{black}{EMM} $\wt{\mQ}$ given by $\wt{\th}^{*}(t)$
and $\wt{\l}_{1}^{*}(t),...,\wt{\l}_{n-1}^{*}(t)$.

\textcolor{black}{Letting} $\ov{\l}^{*}(t)\coloneqq\wt{\l}_{1}^{*}(t)+...+\wt{\l}_{n-1}^{*}(t)$,
this corresponds to the Radon-Nikodym process \textcolor{black}{in
Girsanov's} theorem given by 
\begin{align*}
dZ_{i}(t)=Z_{i}(t-)\int_{\wh{B}}(\varphi_{t}(y)-1)\left[\P_{i}(dy,dt)-\wh{\l}_{t}(dy)dt\right],\quad i=1,...,n,
\end{align*}
where $\varphi_{t}(y)=\frac{\ov{\l}_{t}^{*}(dy)}{\wh{\l}_{t}(dy)}$.
The \textcolor{black}{solution is given by} $Z_{i}(t)=\prod_{j=1}^{\wt{N}(t)}\varphi_{T_{j}}(\wt{Y}_{ij})e^{\int_{0}^{t}\int_{\wh{B}}(1-\varphi_{s}(y))\wh{\l}_{s}(dy)ds},$ where $\{T_{n}\}$ are the jump times and $\{\wt{Y}_{in}\}$ are the
marks of $\P_{i}$. It can also be written as 
\begin{align*}
Z_{i}(t)=e^{\int_{0}^{t}\int_{\wh{B}}\ln\varphi_{s}(y)\P_{i}(dy,ds)+\int_{0}^{t}\int_{\wh{B}}(1-\varphi_{s}(y))\wh{\l}_{s}(dy)ds}.
\end{align*}

Letting $Z(t)=\prod_{i=1}^{n}Z_{i}(t)$ and assuming that $\mE_{\mP}Z(T)=1$,
we know that under the measure defined by $Z$, $\P_{i}(dy,dt)$ has
the intensity $\ov{\l}_{t}^{*}(dy)$.

\bigskip{}

{ \centering \textbf{The Uplifted EMM} \\
 } \bigskip{}

\textcolor{black}{We uplift the EMM in the fictitious market to the
original market in the usual }\textcolor{black}{way to obtain}\textcolor{black}{{}
the following theorem.} 
\begin{thm}
\textcolor{black}{(A Unique Consistent Uplifted EMM)} Assume the reduced
market obtained by complete neglect is arbitrage-free and complete.
Then, there exists a unique consistent uplifted EMM $\mQ^{*}$ of
the unique EMM $\wt{\mQ}$ \textcolor{black}{from }the\textcolor{black}{{}
fictitious market} given by $\wt{\th}^{*}(t)$, $\wt{\l}_{1}^{*}(t),...,\wt{\l}_{n-1}^{*}(t),$
and $\wt{\l}_{m}^{*}(t)=\l_{m}(t),m=n,...,M$. 
\end{thm}

\textcolor{black}{Letting} $\wt{\l}^{*}(t)\coloneqq\wt{\l}_{1}^{*}(t)+...+\wt{\l}_{M}^{*}(t)$,
this corresponds to the Radon-Nikodym process in \textcolor{black}{Girsanov's
theorem }given by 
\begin{align*}
dZ_{i}(t)=Z_{i}(t-)\int_{B}(\varphi_{t}(y)-1)\left[\Psi_{i}(dy,dt)-\l_{t}(dy)dt\right],\quad i=1,...,n,
\end{align*}
where $\varphi_{t}(y)=\frac{\wt{\l}_{t}^{*}(dy)}{\l_{t}(dy)}$. \textcolor{black}{The
solution is} $Z_{i}(t)=\prod_{j=1}^{N(t)}\varphi_{T_{j}}(Y_{ij})e^{\int_{0}^{t}\int_{B}(1-\varphi_{s}(y))\l_{s}(dy)ds},$ where $\{T_{n}\}$ are the jump times and $\{Y_{in}\}$ are the marks
of $\Psi_{i}$. Again, it can also be written as 
\begin{align*}
Z_{i}(t)=e^{\int_{0}^{t}\int_{B}\ln\varphi_{s}(y)\Psi_{i}(dy,ds)+\int_{0}^{t}\int_{B}(1-\varphi_{s}(y))\l_{s}(dy)ds}.
\end{align*}

\begin{rem}
\textcolor{black}{\emph{(The Financial Meaning of Consistency)}}

\textcolor{black}{Consider the Poisson processes whose associated
risks are }\textcolor{black}{neglected in}\textcolor{black}{{} the construction
of the fictitious market. Their intensities $\wt{\l}_{j}$ under the
EMM in the original market have the usual interpretation }\textcolor{black}{of
differing from $\l_{j}$ by a jump risk premium}\textcolor{black}{.
We see that in }\textcolor{black}{the consistent }\textcolor{black}{uplifted
EMM, the equality $\wt{\l}_{m}=\l_{m},m=n-D+1,...,M,$ states that
}\textbf{\textcolor{black}{these particular }}\textcolor{black}{intensities}\textbf{\textcolor{black}{{}
have no risk }}\textcolor{black}{premium. In }\textcolor{black}{essence,
because the trader does not consider }\textcolor{black}{these risks
when }\textcolor{black}{partially hedging any derivative's payoff,
}\textcolor{black}{these jump risks are not adjusted for a risk premium,
i.e. they represent non-priced risks.}
\end{rem}

\section{Filtration Reduction with Partial Neglect}

\textcolor{black}{This section introduces }\textcolor{black}{an alternative}\textcolor{black}{{}
filtration reduction that a trader might choose to employ, which we
called partial neglect.}

\subsection{\textcolor{black}{Discrete Jump Sizes}}

\textcolor{black}{We now consider the scenario where a trader does
not wish to completely ignore the information coming from some subset
}\textcolor{black}{of random processes,}\textcolor{black}{{} but chooses
a reduction that retains some of }\textcolor{black}{their randomness
}\textcolor{black}{as well}. \textcolor{black}{For clarity, we consider
}only one Brownian motio\textcolor{black}{n and focus on the }\textcolor{black}{jump
processes}.

\subsubsection{Trinomial\textcolor{black}{{} to a Binomial Jump} Reduction}

{ \centering \textbf{The Original Market} \\
 } \bigskip{}

\textcolor{black}{This section constructs the original market. }Let
the random jumps be modeled by the following compound Poisson processes.
\begin{align*}
Q_{i}(t)=y_{i,1}N_{1}+y_{i,2}N_{2}(t)+y_{i,3}N_{3}(t),\quad i=1,2,3,
\end{align*}
where $y_{i,m}>-1$ for $i=1,2,3$, $m=1,2,3$, with $\b_{i}=\frac{1}{\l}(\l_{1}y_{i,1}+\l_{2}y_{i,2}+\l_{3}y_{i,3}),\quad i=1,2,3,$ and the market price of risk equations are 
\begin{align*}
\a_{1}-r & =\s_{1}\th+(\l_{1}-\wt{\l}_{1})y_{1,1}+(\l_{2}-\wt{\l}_{2})y_{1,2}+(\l_{3}-\wt{\l}_{3})y_{1,3},\\
\a_{2}-r & =\s_{2}\th+(\l_{1}-\wt{\l}_{1})y_{2,1}+(\l_{2}-\wt{\l}_{2})y_{2,2}+(\l_{3}-\wt{\l}_{3})y_{2,3},\\
\a_{3}-r & =\s_{3}\th+(\l_{1}-\wt{\l}_{1})y_{3,1}+(\l_{2}-\wt{\l}_{2})y_{3,2}+(\l_{3}-\wt{\l}_{3})y_{3,3}.
\end{align*}
\textcolor{black}{Hence, the }market is incomplete.

\bigskip{}

{ \centering \textbf{The Fictitious Market} \\
 } \bigskip{}

\textcolor{black}{Consider the fictitious market corresponding to
the reduced} filtration $\mF$ generated by $N_{1}$ and $N_{2}+N_{3}$.
\textcolor{black}{As before, this explains the term ``partial neglect,''
whereby the }\textcolor{black}{trader exactly hedges the aggregated
risks}\textcolor{black}{{} generated by $N_{2}$ and $N_{3}$.} The
\textcolor{black}{modified} compound Poisson processes are 
\begin{align*}
\ov{Q}_{i}(t) & \coloneqq\mE[Q_{i}(t)|\cF_{t}]\\
 & =\mE[y_{i,1}N_{1}(t)+y_{i,2}N_{2}(t)+y_{i,3}N_{3}(t)|\cF_{t}]\\
 & =y_{i,1}N_{1}(t)+y_{i,2}\mE[N_{2}(t)|\cF_{t}]+y_{i,3}\mE[N_{3}(t)|\cF_{t}]\\
 & =y_{i,1}N_{1}(t)+y_{i,2}\frac{\l_{2}}{\l_{2}+\l_{3}}(N_{2}(t)+N_{3}(t))+y_{i,3}\frac{\l_{3}}{\l_{2}+\l_{3}}(N_{2}(t)+N_{3}(t))\\
 & =y_{i,1}N_{1}(t)+\left(y_{i,2}\frac{\l_{2}}{\l_{2}+\l_{3}}+y_{i,3}\frac{\l_{3}}{\l_{2}+\l_{3}}\right)(N_{2}(t)+N_{3}(t))\\
 & =y_{i,1}N_{1}(t)+\left(y_{i,2}\frac{p(y_{2})}{p(y_{2})+p(y_{3})}+y_{i,3}\frac{p(y_{3})}{p(y_{2})+p(y_{3})}\right)\ov{N}_{2}(t)\\
 & =y_{i,1}N_{1}(t)+\ov{y}_{i,2}\ov{N}_{2}(t).
\end{align*}
\textcolor{black}{Here, the n}ew compound Poisson processes \textcolor{black}{have
}two-valued jump sizes $\ov{Y}_{ij}$, i.e. $y_{i,1}$ and $\ov{y}_{i,2}$.
However, $\ov{y}_{i,2}$ is a convex combination of $y_{i,2}$ and
$y_{i,3}$, which \textcolor{black}{corresponds to the average randomness
generated by }the two jump processes.

The market price of risk equations are 
\begin{align*}
\a_{1}-r & =\s_{1}\th+(\l_{1}-\wt{\l}_{1})y_{1,1}+(\g-\ov{\g})\ov{y}_{1,2},\\
\a_{2}-r & =\s_{2}\th+(\l_{1}-\wt{\l}_{1})y_{2,1}+(\g-\ov{\g})\ov{y}_{2,2},\\
\a_{3}-r & =\s_{3}\th+(\l_{1}-\wt{\l}_{1})y_{3,1}+(\g-\ov{\g})\ov{y}_{3,2},
\end{align*}
where $\g\coloneqq\l_{2}+\l_{3}$ is the intensity of $\ov{N}_{2}\coloneqq N_{2}+N_{3}$.
We have three equations with three unknowns $\th,\wt{\l}_{1},\ov{\g}$.
If there exists a unique solution $(\th^{*},\wt{\l}_{1}^{*},\ov{\g}^{*})$
to the market price of risk equations, then the fictitious market
is arbitrage-free and complete.

\bigskip{}

{ \centering \textbf{The Uplifted EMM } \\
 } \bigskip{}

To uplift the EMM \textcolor{black}{from} the fictitious market to the
original market, \textcolor{black}{we equate} the two systems of market
price of risk equations (the original one and the reduced) and obtain
\begin{align*}
(\l_{2}-\wt{\l}_{2})y_{1,2}+(\l_{3}-\wt{\l}_{3})y_{1,3} & =(\g-\ov{\g})\ov{y}_{1,2},\\
(\l_{2}-\wt{\l}_{2})y_{2,2}+(\l_{3}-\wt{\l}_{3})y_{2,3} & =(\g-\ov{\g})\ov{y}_{2,2},\\
(\l_{2}-\wt{\l}_{2})y_{3,2}+(\l_{3}-\wt{\l}_{3})y_{3,3} & =(\g-\ov{\g})\ov{y}_{3,2}.
\end{align*}
\textcolor{black}{Using} $\g\ov{y}_{i,2}=(\l_{2}+\l_{3})\ov{y}_{i,2}=\l_{2}y_{i,2}+\l_{3}y_{i,3}$,
this yields 
\begin{align*}
\ov{\g}\ov{y}_{1,2} & =\wt{\l}_{2}y_{1,2}+\wt{\l}_{3}y_{1,3},\\
\ov{\g}\ov{y}_{2,2} & =\wt{\l}_{2}y_{2,2}+\wt{\l}_{3}y_{2,3},\\
\ov{\g}\ov{y}_{3,2} & =\wt{\l}_{2}y_{3,2}+\wt{\l}_{3}y_{3,3},
\end{align*}
and since $\ov{y}_{i,2}=y_{i,2}\frac{\l_{2}}{\l_{2}+\l_{3}}+y_{i,3}\frac{\l_{3}}{\l_{2}+\l_{3}}=y_{i,2}\frac{p(y_{2})}{p(y_{2})+p(y_{3})}+y_{i,3}\frac{p(y_{3})}{p(y_{2})+p(y_{3})}=y_{i,2}\d+y_{i,3}(1-\d)$,
\textcolor{black}{letting} $\d\coloneqq\frac{p(y_{2})}{p(y_{2})+p(y_{3})}=\frac{\l_{2}}{\l_{2}+\l_{3}}$,
the required uplift is given by 
\begin{align*}
\wt{\l}_{2}^{*} & =\ov{\g}^{*}\d,\\
\wt{\l}_{3}^{*} & =\ov{\g}^{*}(1-\d).
\end{align*}
We have thus proved an important special case of the desired result. 
\begin{thm}
\textcolor{black}{(A Unique Uplifted EMM) }Assume that in the fictitious
market obtained by partial neglect there is a unique EMM defined by
the vector $(\th^{*},\wt{\l}_{1}^{*},\ov{\g}^{*})$. Then, there is
a unique uplift of the unique EMM from\textcolor{black}{{} the fictitious
market to the original market p}rovided by the vector $(\th^{*},\wt{\l}_{1}^{*},\ov{\g}^{*}\d,\ov{\g}^{*}(1-\d))$,
where $\d\coloneqq\frac{\l_{2}}{\l_{2}+\l_{3}}$. 
\end{thm}

\subsubsection{Multinomial Jumps to a Multinomial Reduction}

{ \centering \textbf{The Original Market} \\
 } \bigskip{}

\textcolor{black}{In the original market, we} assume that there are
$n$ stocks \textcolor{black}{trading and there are} $D$ Brownian
motions and $M$ Poisson processes, where $D<n<M$.\textcolor{cyan}{{}
}\textcolor{black}{The random jumps follow compound}\textcolor{cyan}{{}
}Poisson processes, 
\begin{align*}
Q_{i}(t)=y_{i,1}N_{1}(t)+y_{i,2}N_{2}(t)+...+y_{i,M}N_{M}(t),\quad i=1,...,n,
\end{align*}
where $y_{i,m}>-1$ for $i=1,,,,,n$, $m=1,...,M$, \textcolor{black}{with} $\b_{i}=\frac{1}{\l}(\l_{1}y_{i,1}+\l_{2}y_{i,2}+...+\l_{M}y_{i,M}),\quad i=1,...,M,$ and the market price of risk equations are 
\begin{align*}
\a_{1}-r & =\s_{11}\th_{1}+...+\s_{1D}\th_{D}+(\l_{1}-\wt{\l}_{1})y_{1,1}+...+(\l_{M}-\wt{\l}_{M})y_{1,M},\\
...\\
\a_{n}-r & =\s_{n1}\th_{1}+...+\s_{nD}\th_{D}+(\l_{1}-\wt{\l}_{1})y_{n,1}+...+(\l_{M}-\wt{\l}_{M})y_{n,M}.
\end{align*}
\textcolor{black}{This market is incomplete, because }we have $n$
equations with $D+M$ unknowns $\th_{1},...,\th_{D},\wt{\l}_{1},...,\wt{\l}_{M}$.

\bigskip{}

{ \centering \textbf{The Fictitious Market} \\
 } \bigskip{}

\textcolor{black}{To construct a complete fictitious market,} it is
necessary to reduce the number of jump sizes from $M$ to $n-D$.
Without loss of generality, we reduce \textcolor{black}{to the filtration
$\mF$ generated }by the \textcolor{black}{jump processes} $N_{1},...,N_{n-D-1}$
and the \textcolor{black}{``aggregate'' }$N_{n-D}+...+N_{M}$.

The reduced compound Poisson processes are 
\begin{align*}
\ov{Q}_{i}(t) & \coloneqq\mE[Q_{i}(t)|\cF_{t}]\\
 & =\mE[y_{i,1}N_{1}(t)+...+y_{i,M}N_{M}(t)|\cF_{t}]\\
 & =y_{i,1}N_{1}(t)+...+y_{i,n-D-1}N_{n-D-1}(t)+y_{i,n-D}\mE[N_{n-D}(t)|\cF_{t}]+...+y_{i,M}\mE[N_{M}(t)|\cF_{t}]\\
 & =y_{i,1}N_{1}(t)+...+y_{i,n-D-1}N_{n-D-1}(t)+y_{i,n-D}\mE[N_{n-D+1}(t)|\cF_{t}]+...+y_{i,M}\mE[N_{M}(t)|\cF_{t}]\\
 & =y_{i,1}N_{1}(t)+...+y_{i,n-D-1}N_{n-D-1}(t)+(y_{i,n-D}\d_{n-D}+...+y_{i,M}\d_{M})\ov{N}_{n-D}(t)\\
 & =y_{i,1}N_{1}(t)+...+y_{i,n-D-1}N_{n-D-1}(t)+\ov{y}_{i,n-D}\ov{N}_{n-D}(t),
\end{align*}
where $\d_{i}=\frac{\l_{i}}{\sum_{j=n-D}^{M}\l_{j}}=\frac{p(y_{i})}{\sum_{j=n-D}^{M}p(y_{j})}$,
$\ov{N}_{n-D}(t)=N_{n-D}(t)+...+N_{M}(t)$ and $\ov{y}_{i,n-D}=y_{i,n-D}\d_{n-D}+...+y_{i,M}\d_{M}.$

Again, the\textcolor{black}{{} ``aggregate''} compound Poisson process
corresponds \textcolor{black}{to a} jump size $\ov{y}_{i,n-D}$ that
is a convex combination of the previous \textcolor{black}{values. The}
market price of risk equations \textcolor{black}{in the fictitious
market ar}e 
\begin{align*}
\a_{1}-r & =\s_{11}\th_{1}+...+\s_{1D}\th_{D}+(\l_{1}-\wt{\l}_{1})y_{1,1}+...+(\l_{n-D-1}-\wt{\l}_{n-D-1})y_{1,n-D-1}+(\g-\ov{\g})\ov{y}_{1,n-D},\\
...\\
\a_{n}-r & =\s_{n1}\th_{1}+...+\s_{nD}\th_{D}+(\l_{1}-\wt{\l}_{1})y_{3,1}+...+(\l_{n-D-1}-\wt{\l}_{n-D-1})y_{n,n-D-1}+(\g-\ov{\g})\ov{y}_{n,n-D},
\end{align*}
where $\g\coloneqq\l_{n-D}+...+\l_{M}$ is the intensity of $\ov{N}_{n-D}$
and $\ov{\g}\coloneqq\wt{\l}_{n-D}+...+\wt{\l}_{M}$. \textcolor{black}{This
yields} $n$ equations with $n$ unknowns $\th_{1},...,\th_{D},\l_{1},...,\l_{n-D-1},\ov{\g}$.
If there exists a unique solution $(\th_{1}^{*},...,\th_{D}^{*},\l_{1}^{*},...,\l_{n-D-1}^{*},\ov{\g}^{*})$
to the market price of risk equations, then the fictitious market
is arbitrage-free and complete.

\bigskip{}

{ \centering \textbf{The Uplifted EMM } \\
 } \bigskip{}

\textcolor{black}{We want to uplift the EMM in the fictitious market
to an EMM in the original market.} Equating the two systems of market
price of risk equations (the original \textcolor{black}{and the fictitious}),
we obtain 
\begin{align*}
(\l_{n-D}-\wt{\l}_{n-D})y_{1,n-D}+...+(\l_{M}-\wt{\l}_{M})y_{1,M} & =(\g-\ov{\g})\ov{y}_{1,n-D},\\
 & ...\\
(\l_{n-D}-\wt{\l}_{n-D})y_{n,n-D}+...+(\l_{M}-\wt{\l}_{M})y_{n,M} & =(\g-\ov{\g})\ov{y}_{n,n-D}.
\end{align*}
Using $\g\ov{y}_{i,n-D}=(\l_{n-D}+...+\l_{M})\ov{y}_{i,n-D}=\l_{n-D}y_{i,n-D}+...+\l_{M}y_{i,M}$\textcolor{black}{\ yields}
\begin{align*}
\ov{y}_{1,n-D} & =\frac{\wt{\l}_{n-D}}{\ov{\g}}y_{1,n-D}+...+\frac{\wt{\l}_{M}}{\ov{\g}}y_{1,M},\\
 & ...\\
\ov{y}_{M,n-D} & =\frac{\wt{\l}_{n-D}}{\ov{\g}}y_{n,n-D}+...+\frac{\wt{\l}_{M}}{\ov{\g}}y_{n,M},
\end{align*}
and we immediately recognize \textcolor{black}{that there are} (infinitely)
many convex combinations \textcolor{black}{that produce the} given weighted
\textcolor{black}{average, and there }\textcolor{black}{is no unique
uplift. }\textcolor{black}{We employ the argument used in the discrete
time }trinomial to binomial case\textcolor{black}{{} (see Grigorian and
Jarrow 2024 \cite{key-12,key-13,key-14})}. \textcolor{black}{To obtain
uniqueness, we require that our} uplift be consistent with all the
ordered filtrations (corresponding to consecutive outcomes, i.e. jump
sizes). 
\begin{thm}
\textcolor{black}{(A Unique Consistent Uplifted EMM) }

There is a unique uplift of $(\th_{1}^{*},...,\th_{D}^{*},\l_{1}^{*},...,\l_{n-D-1}^{*},\ov{\g}^{*})$
to the original market that is consistent with all ordered f\textcolor{black}{iltration}
reductions. The unique\textcolor{black}{{} consistent uplifted }\textcolor{black}{EMM
}is given by $(\th_{1}^{*},...,\th_{D}^{*},\l_{1}^{*},...,\l_{n-D-1}^{*},\ov{\g}^{*}\d_{n-D},...,\ov{\g}^{*}\d_{M})$,
where $\d_{i}=\frac{\l_{i}}{\sum_{j=n-D}^{M}\l_{j}}=\frac{p(y_{i})}{\sum_{j=n-D}^{M}p(y_{j})}$
for $i=n-D,...,M$. 
\end{thm}

\begin{proof}
Using a partition $\{y_{i,n-D},y_{i,n-D+1}\}\cup\{y_{i,n-D+2},...,y_{i,M}\}$,
we get for each $i=1,...,n,$ 
\begin{align*}
\frac{\wt{\l}_{n-D}}{\wt{\l}_{n-D}+\wt{\l}_{n-D+1}}y_{i,n-D}+\frac{\wt{\l}_{n-D+1}}{\wt{\l}_{n-D}+\wt{\l}_{n-D+1}}y_{i,n-D+1} & =\frac{\l_{n-D}}{\l_{n-D}+\l_{n-D+1}}y_{i,n-D}+\frac{\l_{n-D+1}}{\l_{n-D}+\l_{n-D+1}}y_{i,n-D+1},\\
\frac{\wt{\l}_{n-D+2}-\l_{n-D+2}}{\sum_{j=n-D+2}^{M}\wt{\l}_{j}}y_{i,n-D+2} & +...+\frac{\wt{\l}_{M}-\l_{M}}{\sum_{j=n-D+2}^{M}\wt{\l}_{j}}y_{i,M}=0.
\end{align*}
This gives us $\wt{\l}_{n-D}=\l_{n-D}$ and $\wt{\l}_{n-D+1}=\l_{n-D+1}$.
Given these values, consider the partition $\{y_{i,n-D},y_{i,n-D+1},y_{i,n-D+2}\}\cup\{y_{i,n-D+3},...,y_{i,M}\}$,
which produces the equations 
\begin{align*}
\frac{\wt{\l}_{n-D+2}}{\l_{n-D}+\l_{n-D+1}+\l_{n-D+2}}y_{i,n-D+2} & =\frac{\l_{n-D+2}}{\l_{n-D}+\l_{n-D+1}+\l_{n-D+2}}y_{i,n-D+2},\\
\frac{\wt{\l}_{n-D+3}-\l_{n-D+3}}{\sum_{j=n-D+3}^{M}\wt{\l}_{j}}y_{i,n-D+3} & +...+\frac{\wt{\l}_{M}-\l_{M}}{\sum_{j=n-D+3}^{M}\wt{\l}_{j}}y_{i,M}=0,
\end{align*}
from which we get $\wt{\l}_{n-D+2}=\l_{n-D+2}$. The remaining values
are recovered similarly. 
\end{proof}

\subsubsection{A General \textquotedblleft Batching\textquotedblright{} \textcolor{black}{Procedure }}

From the analysis above it is clear that \textcolor{black}{the }choice
of aggregation \textcolor{black}{can be generalized to arbitrary}
batches of Poisson processes that the trader wishes\textcolor{black}{{}
to exactly hedge on average}. \textcolor{black}{Indeed, }in a market
with $n$ stocks, $D$ Brownian motions, and $M$ Poisson processes,
where $D<n<M$, we \textcolor{black}{need} to reduce $M$ to $n-D$.
For that purpose, partition $\{1,...,M\}$ into $n-D$ subsets $\{1...,K_{1}\}$,$\{K_{1}+1,...,K_{1}+K_{2}\}$,...,$\{\sum_{j=1}^{n-D-1}K_{j}+1,...,M\}$
and project the compound Poisson processes onto the filtrations generated
by the $n-D$ aggregated processes $N_{1}+...+N_{K_{1}},...,N_{\sum_{j=1}^{n-D-1}K_{j}+1}+...+N_{M}$.
If this produces an arbitrage-free and complete market, we obtain
a unique EMM $\wt{\mQ}$ \textcolor{black}{in the fictitious market}
and uplift it\textcolor{black}{{} to the unique consistent EMM in the
original market} via the procedure outlined above. Because the notation
gets rather complicated and no additional insights emerge, we omit
the details.

\subsection{Time-Varying Coefficients and Inhomogeneous Poisson Processes}

\textcolor{black}{We now consider the case where the Poisson processes
have deterministic intensities}\textcolor{black}{{} where the}\textcolor{black}{{}
mean rate of return $\a(t)$, the money market rate $r(t)$}\textcolor{black}{,
and the }\textcolor{black}{volatility $\s(t)$ are also deterministic,
with $\s(t)$ being square integrable. }

\subsubsection{The Original Market}

\textcolor{black}{In the original market, we} assume there are $n$
stocks trading, one Brownian motion, and $M$ Poisson processes with
deterministic intensity functions $\l_{1}(t)$,...,$\l_{M}(t)$, where
$n<M$. The restriction to one Brownian motion is \textcolor{black}{for
expositional clarity} and without loss of generality. The filtration
$\mG$ is the natural one generated by the Brownian motion and Poisson
processes. \textcolor{black}{The original market's stocks evolve a}s
\begin{align*}
dS_{i}(t)=\left(\a_{i}(t)+\int_{B}y\l_{t}(dy)\right)S_{i}(t)dt+S_{i}(t)\s_{i}(t)dW(t)+S_{i}(t-)\int_{B}y\wt{\Psi}_{i}(dy,dt),\quad i=1,...,n,
\end{align*}
where $B\coloneqq\{y_{1},...,y_{M}\}\subset(-1,\infty)$ is the mark
space, $\l_{t}(\cdot)$ is the measure on the mark space defined via
$\l_{t}(dy)=\l_{t}(B)\G_{t}(dy)$, with $\G_{t}(dy)$ being the distribution
of the jump size given an event at time $t$, and $\wt{\Psi}_{i}(dy,dt)=\Psi_{i}(dy,dt)-\l_{t}(dy)dt,$ where $\Psi_{i}(dy,dt)$ is the random measure that governs the jumps.
As before, the integrals are in fact sums.

The market price of risk equations are 
\begin{align*}
\a_{1}(t)-r(t) & =\s_{1}(t)\th(t)+(\l_{1}(t)-\wt{\l}_{1}(t))y_{1,1}+...+(\l_{M}(t)-\wt{\l}_{M}(t))y_{1,M},\\
...\\
\a_{n}(t)-r(t) & =\s_{n}(t)\th(t)+(\l_{1}(t)-\wt{\l}_{1}(t))y_{n,1}+...+(\l_{M}(t)-\wt{\l}_{M}(t))y_{n,M},
\end{align*}
and the market is incomplete.

\subsubsection{The Fictitious Market}

\textcolor{black}{To construct the complete fictitious market, we
reduce to the filtration generated by the Brownian motion, the Poisson
processes $N_{1},...,N_{n-2}$ and} the \textcolor{black}{``aggregate''}
$N_{n-1}+...+N_{M}$. The modified jump processes are given by 
\begin{align*}
d\ov{Q}_{i}(t) & =y_{i,1}dN_{1}(t)+...+y_{i,n-2}dN_{n-2}(t)+(y_{i,n-1}\d_{n-1}(t)+...+y_{i,M}\d_{M}(t))d\ov{N}_{n-D}(t)\\
 & =y_{i,1}dN_{1}(t)+...+y_{i,n-2}dN_{n-2}(t)+\ov{y}_{i,n-1}(t)d\ov{N}_{n-1}(t),
\end{align*}
where $\d_{k}(t)=\frac{\l_{k}(t)}{\sum_{j=n-1}^{M}\l_{j}(t)},k=n-1,...,M$,
$\ov{N}_{n-1}(t)=N_{n-1}(t)+...+N_{M}(t)$ and $\ov{y}_{i,n-1}(t)=y_{i,n-1}\d_{n-1}(t)+...+y_{i,M}\d_{M}(t).$

The jumps of the reduced market are controlled by the compensated random measure $\wt{\P}_{i}(dy,dt)=\P_{i}(dy,dt)-\wh{\l}_{t}(dy)dt,$ \textcolor{black}{with} $\int_{\wh{B}}y\P_{i}(dy,dt)=\sum_{T_{n}\leq t}\wt{Y}_{T_{n}}^{i}$\textcolor{black}{,
and} on $\{T_{n}=t\}$ the distribution of $\wt{Y}_{i,t}$ is given
by $\wt{\G}_{t}(dy)=\frac{\wh{\l}_{t}(dy)}{\wh{\l}_{t}(\wh{B}_{t})}$,
where on each set $\{T_{n}=t\}$ the intensity measure $\wh{\l}_{t}(\cdot)$
is distributed on $\wh{B}_{t}=\{y_{i,1},...,y_{i,n},\ov{y}_{i,n-1}(t)\}$
via the assignment $\{\l_{1},...,\l_{n-2},\l_{n-1}+...+\l_{M}\}$.

\textcolor{black}{The market price of risk equations are} 
\begin{align*}
\a_{1}(t)-r(t) & =\s_{1}(t)\th(t)+(\l_{1}(t)-\wt{\l}_{1}(t))y_{1,1}+...+(\g(t)-\ov{\g}(t))\ov{y}_{1,n-1}(t),\\
...\\
\a_{n}(t)-r(t) & =\s_{n}(t)\th(t)+(\l_{1}(t)-\wt{\l}_{1}(t))y_{n,1}+...+(\g(t)-\ov{\g}(t))\ov{y}_{n,n-1}(t).
\end{align*}
\textcolor{black}{There are} $n$ equations in $n$ unknowns $\th(t),\wt{\l}_{1}(t),...,\wt{\l}_{n-2}(t),\ov{\g}(t)$.
Assuming there exists a unique solution, we obtain an arbitrage-free
and complete market, and hence a unique EMM $\wt{\mQ}$ \textcolor{black}{in
the fictitious market }given by $\wt{\th}^{*}(t)$ and $\wt{\l}_{1}^{*}(t),...,\wt{\l}_{n-2}^{*}(t),\ov{\g}^{*}(t)$.

\textcolor{black}{Letting} $\ov{\l}^{*}(t)\coloneqq\wt{\l}_{1}^{*}(t)+...+\wt{\l}_{n-2}^{*}(t)+\ov{\g}^{*}(t)$,
this corresponds to the Radon-Nikodym process in \textcolor{black}{Girsanov's
theorem} given by 
\begin{align*}
dZ_{i}(t)=Z_{i}(t-)\int_{\wh{B}}(\varphi_{t}(y)-1)\left[\P_{i}(dy,dt)-\wh{\l}_{t}(dy)dt\right],\quad i=1,...,n,
\end{align*}
where $\varphi_{t}(y)=\frac{\ov{\l}_{t}^{*}(dy)}{\wh{\l}_{t}(dy)}$.
The solution is $Z_{i}(t)=\prod_{j=1}^{\wt{N}(t)}\varphi_{T_{j}}(\wt{Y}_{ij})e^{\int_{0}^{t}\int_{\wh{B}}(1-\varphi_{s}(y))\wh{\l}_{s}(dy)ds},$ where $\{T_{n}\}$ are the jump times and $\{\wt{Y}_{in}\}$ are the
marks of $\P_{i}$. It can also be written as 
\begin{align*}
Z_{i}(t)=e^{\int_{0}^{t}\int_{\wh{B}}\ln\varphi_{s}(y)\P_{i}(dy,ds)+\int_{0}^{t}\int_{\wh{B}}(1-\varphi_{s}(y))\wh{\l}_{s}(dy)ds}.
\end{align*}
Letting $Z(t)=\prod_{i=1}^{n}Z_{i}(t)$ and assuming that $\mE_{\mP}Z(T)=1$,
we know that under the measure defined by $Z$, $\P_{i}(dy,dt)$ has
the required intensity $\ov{\l}_{t}^{*}(dy)$.

\bigskip{}

{ \centering \textbf{The Uplifted EMM} \\
 } \bigskip{}

\textcolor{black}{Uplifting the EMM in the fictitious market to the
original market in the usual manner, we can prove the following theorem.} 
\begin{thm}
\textcolor{black}{(A Unique Consistent Uplifted EMM) }Assume the \textcolor{black}{fictitious}
market obtained by partial neglect is arbitrage-free and complete.
Then, there exists a unique consistent \textcolor{black}{uplifted
EMM} $\mQ^{*}$ \textcolor{black}{in the original market} \textcolor{black}{from
}the unique EMM $\wt{\mQ}$ in \textcolor{black}{the fictitious market}
given by $\wt{\th}^{*}(t)$, $\wt{\l}_{1}^{*}(t),...,\wt{\l}_{n-2}^{*}(t),$
and $\wt{\l}_{m}^{*}(t)=\ov{\g}^{*}(t)\d_{m}(t),m=n-1,...,M$, where
$\d_{m}(t)=\frac{\l_{m}(t)}{\sum_{j=n-1}^{M}\l_{j}(t)}$. 
\end{thm}

\textcolor{black}{Letting} $\wt{\l}^{*}(t)\coloneqq\wt{\l}_{1}^{*}(t)+...+\wt{\l}_{M}^{*}(t)$,
this corresponds to the Radon-Nikodym process in Girsanov's theorem
given by 
\begin{align*}
dZ_{i}(t)=Z_{i}(t-)\int_{B}(\varphi_{t}(y)-1)\left[\Psi_{i}(dy,dt)-\l_{t}(dy)dt\right],\quad i=1,...,n,
\end{align*}
where $\varphi_{t}(y)=\frac{\wt{\l}_{t}^{*}(dy)}{\l_{t}(dy)}$. The
solution\textcolor{black}{{} is} $Z_{i}(t)=\prod_{j=1}^{N(t)}\varphi_{T_{j}}(Y_{ij})e^{\int_{0}^{t}\int_{B}(1-\varphi_{s}(y))\l_{s}(dy)ds},$ where $\{T_{n}\}$ are the jump times and $\{Y_{in}\}$ are the marks
of $\Psi_{i}$. Again, it can also be written as 
\begin{align*}
Z_{i}(t)=e^{\int_{0}^{t}\int_{B}\ln\varphi_{s}(y)\Psi(dy,ds)+\int_{0}^{t}\int_{B}(1-\varphi_{s}(y))\l_{s}(dy)ds}.
\end{align*}

\begin{rem}
(\textcolor{black}{\emph{The Financial Meaning of Consistency)}}

\textcolor{black}{Relying again on the interpretation of $\wt{\l}_{j}$
as containing a}\textcolor{black}{{} jump risk premium, the equality}\textcolor{black}{{}
$\wt{\l}_{m}^{*}=\wt{\g}^{*}\d_{m},m=n-D,...,M,$ is equivalent to
saying that the ``averaged'' jump risks receive a reduced risk premium
according to their }\textcolor{black}{weights in the}\textcolor{black}{{}
subset of batched risks. Incidentally, these weights are the conditional
probabilities, as was seen }\textcolor{black}{in the }\textcolor{black}{discrete-time
and Brownian settings }\textcolor{black}{(see Grigorian and Jarrow 2024
\cite{key-12,key-13,key-14}). }
\end{rem}

\subsection{\textcolor{black}{Partial Neglect with }Continuously Distributed
Jump Sizes}

We now extend the \textcolor{black}{model} to the case of continuously
distributed jumps \textcolor{black}{where} the jump sizes $Y_{i}$
have a density $f(y)$ \textcolor{black}{and} the support of the distribution
of $Y_{i}$ is some subset $B$ of $(-1,\infty)$.\textcolor{black}{{}
We assume throughout the }\textcolor{black}{rest of the paper}\textcolor{black}{{}
that in the usual factorization $S_{t}=S_{t}^{\wh{B}}S_{t}^{\wh{B}^{c}}$,
the key integrability condition of lemma \ref{exp mart} holds for
$S_{t}^{\wh{B}^{c}}$, i.e. $\mE_{\mP}[S_{t}^{\wh{B}^{c}}]=1$ for
all $t\in[0,T]$}. Using the notation above, we have that $\b=\mE Y_{i}=\int_{B}yf(y)dy.$ In the simplest case \textcolor{black}{with the original market consisting
of}\textcolor{black}{{} just one}\textcolor{black}{{} stock, }one Brownian
motion, and one Poisson process, the choice of \textcolor{black}{an
EMM }corresponds to choosing $\th,\wt{\l}$ and a density $\wt{f}(y)$
such that the distribution measure is supported on $B$ and the market
price of risk equation for \textcolor{black}{the single stock} 
\begin{align*}
\a-r=\s\th+\b\l-\wt{\b}\wt{\l}
\end{align*}
is satisfied, where $\wt{\b}=\wt{\mE}Y_{i}=\int_{B}y\wt{f}(y)dy.$
\textcolor{black}{This} market is incomplete and \textcolor{black}{there
exists a continuum of EMMs.} \textcolor{black}{The source} of incompleteness
is the arbitrariness in the choice of \textcolor{black}{the density}
$\wt{f}$.

\textcolor{black}{Now, }\textcolor{black}{consider a}\textcolor{black}{{}
market with $n$ stocks trading, $D>n$ Brownian motions, and the
single Poisson process as given above.} This yields the following
set of equations for the market prices of risk: 
\begin{align*}
\a_{1}-r & =\s_{11}\th_{1}+...+\s_{1D}\th_{D}+\b\l-\wt{\b}\wt{\l},\\
 & ...\\
\a_{n}-r & =\s_{n1}\th_{1}+...+\s_{nD}\th_{D}+\b\l-\wt{\b}\wt{\l}.\\
\end{align*}
Applying the machinery of filtration reduction used in the preceding
sections, for each $i=1,...,n,$ \textcolor{black}{we partition $B$
(the support of the distribution of $Y_{i}$) } into $n-D$ subsets
$B_{1}^{i},...,B_{n-D}^{i}$ (different as sets but of equal positive
measure across $i$, i.e. $\l_{t}(B_{j}^{1})=\l_{t}(B_{j}^{2})=...=\l_{t}(B_{j}^{n})$
for all $j=1,...,n-D$). \textcolor{black}{Then, we obtain} partial
averages via 
\begin{align*}
y_{i,1} & =\frac{\int_{B_{1}^{i}}yf(y)dy}{\int_{B_{1}^{i}}f(y)dy},\\
 & ...\\
y_{i,n-D} & =\frac{\int_{B_{n-D}^{i}}yf(y)dy}{\int_{B_{n-D}^{i}}f(y)dy},
\end{align*}
and the corresponding probability assignments via\footnote{We dropped the index $i$ since the measures do not depend on it.}
\begin{align*}
p(y_{1}) & =\int_{B_{1}}f(y)dy,\\
 & ...\\
p(y_{n-D}) & =\int_{B_{n-D}}f(y)dy.
\end{align*}
\textcolor{black}{This gives a} decomposition of the single Poisson
process into $n-D$ thinned Poisson processes via $\l_{m}=\l p(y_{m})$,
$m=1,...,n-D$. \textcolor{black}{This decomposition yields}\textcolor{black}{{}
the market price} of risk equations 
\begin{align*}
\a_{1}-r & =\s_{11}\th_{1}+...+\s_{1D}\th_{D}+(\l_{1}-\wt{\l}_{1})y_{1,1}+...+(\l_{n-D}-\wt{\l}_{n-D})y_{1,n-D},\\
...\\
\a_{n}-r & =\s_{n1}\th_{1}+...+\s_{nD}\th_{D}+(\l_{1}-\wt{\l}_{1})y_{n,1}+...+(\l_{n-D}-\wt{\l}_{n-D})y_{n,n-D}.
\end{align*}
Assuming that this system has a unique solution, denoted by $\th_{1}^{*},...,\th_{D}^{*}$,
$\wt{\l}_{1}^{*},...,\wt{\l}_{n-D}^{*},$ we obtain \textcolor{black}{the
}\textcolor{black}{unique }\textcolor{black}{EMM in the fictitious market. }

\textcolor{black}{To uplift the unique EMM $\wt{\mQ}$ }\textcolor{black}{from
the }\textcolor{black}{fictitious market to the} original market,
we still need to specify\textcolor{black}{{} the density} $\wt{f}$. Note
that we have 
\begin{align*}
\wt{p}^{*}(y_{m}) & =\frac{\wt{\l}_{m}^{*}}{\wt{\l}^{*}},\quad m=1,...,n-D,\\
\wt{\b}_{i}=\sum_{m=1}^{n-D}y_{i,m}\wt{p}^{*}(y_{m}) & =\frac{1}{\wt{\l}^{*}}\sum_{m=1}^{n-D}\wt{\l}_{m}^{*}y_{i,m},\quad i=1,...,n.
\end{align*}
hence the desired $\wt{f}$ must be such that 
\begin{align*}
\wt{f}(y)dy & \sim f(y)dy,\\
\int_{B_{m}}\wt{f}(y)dy & =\wt{p}^{*}(y_{m}),\quad m=1,...,n-D,\\
\frac{\int_{B_{m}}y\wt{f}(y)dy}{\int_{B_{m}}\wt{f}(y)dy} & =y_{m},\quad m=1,...,n-D.
\end{align*}

\textcolor{black}{Although we have} narrowed the collection of candidates
functions for $\wt{f}$, the set is still infinite. The situation
\textcolor{black}{is similar }to the one encountered in the discrete-time
setting with increments having a general distribution on $\mR$. Arguing
in the same vein, the density defined via 
\begin{align*}
\wt{f}^{*}(y)\coloneqq\sum_{m=1}^{n-D}\frac{f(y)\1_{B_{m}}(y)}{\int_{B_{m}}f(y)dy}\wt{p}^{*}(y_{m})
\end{align*}
is the only uplift consistent with all \textcolor{black}{filtration
}reductions. Indeed, while there are many possible functions $\wt{f}$
which satisfy the partial average property above (i.e. agree with
$\wt{p}^{*}(y_{m})$ and $y_{m}$ on $B_{m}$), if we choose a different
partition of $B$, then these functions will be excluded from the
collection. If we insist that the procedure be consistent with all
\textcolor{black}{possible filtration }reductions (i.e. partitions
of $B$), then the above formula becomes the only viable candidate.
\textcolor{black}{This yields the following theorem. } 
\begin{thm}
\textcolor{black}{(A Unique Consistent Uplifted EMM)} Assume the \textcolor{black}{fictitious}
market is arbitrage-free and complete. Then, there exists a unique
consistent uplifted \textcolor{black}{EMM} $\mQ^{*}$ \textcolor{black}{from
the} unique EMM $\wt{\mQ}$ in the\textcolor{black}{{} fictitious market}
given by $\th_{1}^{*},...,\th_{D}^{*}$, $\wt{\l}_{1}^{*},...,\wt{\l}_{n-D}^{*},$
and 
\begin{align*}
\wt{f}^{*}(y)\coloneqq\sum_{m=1}^{n-D}\frac{f(y)\1_{B_{m}}(y)}{\int_{B_{m}}f(y)dy}\wt{p}^{*}(y_{m}).
\end{align*}
\end{thm}

\subsection{The General Theorem }

\textcolor{black}{Finally, we analyze the most general version of
the model. Here, we consider} a general jump-diffusion model with
both time-varying parameters (intensities, volatilities, mean rates
of return, money market rates) and continuously distributed jump sizes.
The fact that $\l_{t}$ is \textcolor{black}{a non-constant deterministic
function of time implies} that the distribution of the jump amplitudes
is also \textcolor{black}{time-varying. Because }we have done most
of the technical work in the prior sections, we only focus on the
key steps of the \textcolor{black}{filtration} reduction procedure.

We assume that we have \textcolor{black}{an original }market with
$n$ stocks, $D$ Brownian motions, and $n$ marked point processes
$\Psi_{i}(dt,dy)$ with the same deterministic intensity function
$\l_{t}(dy)$, where $D<n$. \textcolor{black}{Here, }the incompleteness
comes from the continuously distributed jumps. \textcolor{black}{The
stock dynamics are} 
\begin{align*}
dS_{i}(t)=\left(\a_{i}(t)+\int_{B}y\l_{t}(dy)\right)S_{i}(t)dt+S_{i}(t)\s_{i}(t)dW(t)+S_{i}(t-)\int_{B}y\wt{\Psi}_{i}(dt,dy),\quad i=1,...,n,
\end{align*}
where $B\subset(-1,\infty),i=1,...,n,$ is the mark space, equipped
with the Borel $\s$-algebra $\cB$, $\l_{t}(\cdot)$ is the measure
on the mark space defined via $\l_{t}(dy)=\l_{t}(B)\G_{t}(dy)$, with
$\G_{t}(dy)$ being the distribution of the jump size given an event
at time $t$, and $\wt{\Psi}_{i}(dt,dy)=\Psi_{i}(dt,dy)-\l_{t}(dy)dt,$ where $\Psi_{i}(dt,dy)$ is the random measure that governs the jumps.

\textcolor{black}{To simplify }\textcolor{black}{the notation and to
}\textcolor{black}{focus only on the new ideas,} \textcolor{black}{we
}\textcolor{black}{consider a single} jump process given by $\int_{0}^{t}\int_{B}y(\Psi_{i}(ds,dy)-\l_{s}(dy)ds).$ The \textcolor{black}{filtration reduction} procedure can be summarized
as follows:

(i) For each $i=1,...,n$, partition the mark space $B$ into $n-D$
disjoint $\cB$-measurable subsets $B_{1}^{i},...,B_{n-D}^{i}$ (different
as sets but of equal measure across $i$, i.e. $\l_{t}(B_{j}^{1})=\l_{t}(B_{j}^{2})=...=\l_{t}(B_{j}^{n})$
for all $j=1,...,n-D$ and for all $t$). We will drop the superscript
\textcolor{black}{when unnecessary. This gives} a collection of independent
Poisson processes $N^{B_{1}},...,N^{B_{n-D}}$ with intensities $\l_{t}(B_{1}),...,\l_t(B_{n-D})$,
respectively.

(ii) obtain the partial averages over the sets in the partition via
\begin{align*}
\ov{y}_{i,m}(t)=\frac{\int_{B_{m}^{i}}yf_{t}(y)dy}{\int_{B_{m}^{i}}f_{t}(y)dy},\quad m=1,...,n-D,i=1,...,n,
\end{align*}
where $f_{t}(y)dy=\G_{t}(dy)=\frac{\l_{t}(dy)}{\l_{t}(B)}$. Note
that the partial averages over $B_{m}^{i}$ will in general be different
across $i$. In fact, this must hold for the market price of risk
equations to have a solution.

(iii) Obtain $\wt{\l}_{t}^{*}(B_{m}),m=1,...,n-D,$ from the fictitious
complete market, and hence the probabilities 
\begin{align*}
\wt{\mP}^{*}\left(Y_{it}\in B_{m}^{i}\right)=\wt{\G}_{t}^{*}(B_{m}^{i})=\wt{p}^{*}(\ov{y}_{i,m}(t))=\frac{\wt{\l}_{t}^{*}(B_{m})}{\wt{\l}_{t}^{*}(B)},
\end{align*}
where we have intentionally \textcolor{black}{employed the different
notations }used throughout the text.

(iv) uplift \textcolor{black}{the unique EMM in the fictitious market}
via $\th_{1}^{*}(t),...,\th_{D}^{*}(t),\wt{\l}_{t}^{*}(dy)$ and the
densities 
\begin{align*}
\wt{f}_{t}^{*}(y)\coloneqq\sum_{m=1}^{n-D}\frac{f_{t}(y)\1_{B_{m}}(y)}{\int_{B_{m}}f_{t}(y)dy}\wt{p}^{*}(\ov{y}_{i,m}(t)).
\end{align*}
\textcolor{black}{We finally obtain the} main result. 
\begin{thm}
\textcolor{black}{(A Unique Consistent Uplifted EMM) }Assume a general
jump-diffusion model driven by Brownian motions and a marked point
process with a deterministic intensity function. Let the \textcolor{black}{fictitious
market be} arbitrage-free and complete. Then, there exists a unique
consistent uplifted EMM $\mQ^{*}$ given by (iv) above of the unique
EMM $\wt{\mQ}$ \textcolor{black}{from the}\textcolor{black}{{} fictitious
market}. 
\end{thm}

\section{Pricing and Hedging in the Original Market}

\textcolor{black}{Consider now a $\cG_{T}$-measurable random variable
$C$ in the original market. The unique uplifted consistent EMM $\mQ^{*}$
can be used to price any derivative $C$ in the original market and
obtain }\textcolor{black}{its risk}\textcolor{black}{-neutral value
$\mE_{\mQ^{*}}(C)$. In the following, we assume we are given an EMM
$\wt{\mQ}$ in the fictitious market under which }\textcolor{black}{any
derivative's synthetically constructed value processes are}\textcolor{black}{{}
martingales.}

\textcolor{black}{Let $\wt{C}=\mE_{\mQ^{*}}(C|\cF_{T})$ be a derivative
in the fictitious economy corresponding to $C$, but $\cF_{T}$-measurable.
Because this derivative is in the} fictitious market, it can be synthetically
constructed using the\textcolor{black}{{} price process} $\wt{S}$ in
the fictitious market. Indeed, the filtration $\mF$ is generated
by the (potentially reduced vector of) Brownian motions and the reduced
$\wt{\P}$. Hence, we can apply the general martingale representation
theorem for jump-diffusions. Namely, there exists an admissible trading
strategy $\wt{\p}$ given by $\wt{\a}=\{\wt{\a}(t),t\in[0,T]\}$ and
$\wt{h}=\{\wt{h}_{t}(y)\}$ such that 
\begin{align*}
\wt{C}=\wt{V}_{T}(\wt{\p})=\wt{V}_{0}+\int_{0}^{T}\wt{\a}(t)\cdot d\wt{S}(t)+\int_{0}^{T}\int_{B}\wt{h}_{t}(y)\wt{\P}(dt,dy),
\end{align*}
where we have used the notation\textcolor{black}{{} }$\a(t)\cdot dS(t)\coloneqq\sum_{j}\a_{j}(t)dS_{j}(t)$.

\subsection{Cost of construction}

Recall that the value process \textcolor{black}{in the fictitious
market} is given by $\wt{C}_{t}=\wt{V}_{t}(\wt{\p})=\mE_{\wt{\mQ}}(\wt{C}|\cF_{t})$.
The cost of construction of the trading \textcolor{black}{strategy
is} 
\begin{align*}
\mE_{\wt{\mQ}}(\wt{C})=\wt{V}_{0}+\mE_{\wt{\mQ}}\left(\int_{0}^{T}\wt{\a}(t)\cdot d\wt{S}(t)\right)+\mE_{\wt{\mQ}}\left(\int_{0}^{T}\int_{B}\wt{h}_{t}(y)\wt{\P}(dt,dy)\right)=\wt{V}_{0},
\end{align*}
since $\int_{0}^{t}\wt{\a}_{s}\cdot d\wt{S}_{s}$ and $\int_{0}^{t}\int_{B}\wt{h}_{s}(y)\wt{\P}(ds,dy)$
are martingales under $\wt{\mQ}$. \textcolor{black}{Given} $\wt{C}$
is $\cF_{T}$-measurable and $\mQ^{*}|_{\cF_{T}}=\wt{\mQ},$ it follows that $\mE_{\wt{\mQ}}(\wt{C})=\mE_{\mQ^{*}}(\wt{C}).$ But $\mE_{\mQ^{*}}(\wt{C})=\mE_{\mQ^{*}}(\mE_{\mQ^{*}}(C|\cF_{T}))=\mE_{\mQ^{*}}(C)$ by the law of iterated expectations, therefore $\mE_{\wt{\mQ}}(\wt{C})=\mE_{\mQ^{*}}(C).$ This means that the cost of construction of $\wt{C}=\mE_{\mQ^{*}}(C|\cF_{T})$
in the fictitious market is the arbitrage-free price for $C$ in the
original market \textcolor{black}{under the consistent uplifted }EMM
$\mQ^{*}.$

\subsection{Non-priced Hedging Error}

In the fictitious market, we know there is an admissible trading strategy
$\wt{\p}$ that generates $\wt{C}=\wt{V}_{T}(\p)=\wt{V}_{0}+\int_{0}^{T}\wt{\a}(t)\cdot d\wt{S}(t)+\int_{0}^{T}\int_{B}\wt{h}_{t}(y)\wt{\P}(dt,dy)$
with the cost of construction $\mE_{\wt{\mQ}}(\wt{C})=\wt{V}_{0}$.
If $\wt{\p}$ is admissible with respect to $S$ and $\wt{\Psi}$
and we use this trading strategy in the original economy with respect
to $C$, it generates \textcolor{black}{the time $T$ payoff }
\begin{align*}
\wt{V}_{0}+\int_{0}^{T}\wt{\a}(t)\cdot dS(t)+\int_{0}^{T}\int_{B}\wt{h}_{t}(y)\wt{\Psi}(dt,dy).
\end{align*}
The \textit{hedging error} is thus 
\begin{align*}
\e(\w)\coloneqq C(\w)-\wt{V}_{0}-\int_{0}^{T}\wt{\a}(t)\cdot dS(t)-\int_{0}^{T}\int_{B}\wt{h}_{t}(y)\wt{\Psi}(dt,dy).
\end{align*}
Taking expectations, we get 
\begin{align*}
\mE_{\mQ^{*}}(\e) & =\mE_{\mQ^{*}}(C)-\mE_{\mQ^{*}}(\wt{V}_{0})-\mE_{\mQ^{*}}\left(\int_{0}^{T}\wt{\a}(t)\cdot dS(t)\right)-\mE_{\mQ^{*}}\left(\int_{0}^{T}\int_{B}\wt{h}_{t}(y)\wt{\Psi}(dt,dy)\right)\\
 & =\mE_{\wt{\mQ}}(\wt{C})-\wt{V}_{0}-\mE_{\mQ^{*}}\left(\int_{0}^{T}\wt{\a}(t)\cdot dS(t)\right)-\mE_{\mQ^{*}}\left(\int_{0}^{T}\int_{B}\wt{h}_{t}(y)\wt{\Psi}(dt,dy)\right),
\end{align*}
where we have used $\mE_{\mQ^{*}}(C)=\mE_{\wt{\mQ}}(\wt{C})$. But,
$\mE_{\wt{\mQ}}(\wt{C})=\wt{V}_{0}$, hence 
\begin{align*}
\mE_{\mQ^{*}}(\e)=-\mE_{\mQ^{*}}\left(\int_{0}^{T}\wt{\a}(t)\cdot dS(t)\right)-\mE_{\mQ^{*}}\left(\int_{0}^{T}\int_{B}\wt{h}_{t}(y)\wt{\Psi}(dt,dy)\right).
\end{align*}
\textcolor{black}{Thus, }under $\mQ^{*}$, the hedging error is not
priced if and only if \\
 $\mE_{\mQ^{*}}\left(\int_{0}^{T}\wt{\a}(t)\cdot dS(t)+\int_{0}^{T}\int_{B}\wt{h}_{t}(y)\wt{\Psi}(dt,dy)\right)=0$.
Various conditions ensure \textcolor{black}{this equality,} but we
do not\textcolor{black}{{} pursue these conditions here}.

\section{Conclusion}

There exists \textcolor{black}{a vast literature on the} different criteria
for choosing an EMM/ELMM \textcolor{black}{to price derivatives} in
an incomplete market. Our method of choosing a unique EMM relies entirely
on an information-based argument. Specifically, a trader chooses a
subset of the information flow that matters to them \textcolor{black}{-
a filtration reduction -} and under these information constraints
obtains a unique EMM. While there are potentially infinitely many
possible \textcolor{black}{filtration }reductions, two distinctly
different approaches\textcolor{black}{{} are studied in this paper}. One
consists in ignoring a particular \textcolor{black}{subset of risks
to obtain a measure} that does \textcolor{black}{not hedge these} risks.
The other consists in\textcolor{black}{{} aggregating} certain sources
\textcolor{black}{of }\textcolor{black}{risks to be hedged on average}\textcolor{black}{{}
to obtain a measure that prices the individual risks in the batch
according to their weight }in the \textcolor{black}{aggregated composite
structure of} risks. These \textcolor{black}{weights are }the conditional
probabilities, which matches with the results obtained in the discrete-time
and Brownian settings.

The key idea of our methodology is that \textbf{information reduction
uniquely determines the equivalent martingale measure} \textcolor{black}{for
pricing} derivatives in an incomplete market. It should be noted,
however, that while in principle every reduction produces a corresponding
measure, it may not be always possible to obtain closed-form expressions
for projected stock prices, density processes, uplifts and other,
sometimes auxiliary, but still important objects. We have focused
our attention on those cases that are most amenable \textcolor{black}{to
these transformations and produce explicit formulas that can be applied
directly in practice.}

\end{document}